\documentclass{eptcs}
\providecommand{\event}{RTRTS 2010} % Name of the event you are submitting to
\usepackage{breakurl}             % Not needed if you use pdflatex only.
\usepackage{alltt}
\usepackage{amsmath,amsthm,amssymb}
\usepackage{enumerate}
\usepackage{graphicx,color}
\usepackage{enumerate}
\usepackage{fancyvrb}

\newcommand{\always}{\Box}
\newcommand{\eventually}{\Diamond}
\newcommand{\until}[1]{\ \mathsf{U}_{#1}\ }
\newcommand{\weakuntil}{\ \mathsf{W}\ }
\newcommand{\Paths}{\mathit{Paths}}

\newtheorem{definition}{Definition}
\newtheorem{theorem}{Theorem}

\newtheorem{lemma}{Lemma}

\begin{SaveVerbatim}{AND}
/\
\end{SaveVerbatim}
\begin{SaveVerbatim}{OR}
\/
\end{SaveVerbatim}

\title{Modeling and Analyzing Adaptive User-Centric Systems in
  Real-Time Maude\thanks{This work has been partially sponsored by the EC project
REFLECT, IST-2007-215893.}} 
  \author{Martin Wirsing\quad \quad Sebastian
  S. Bauer \quad \quad Andreas Schroeder 
  \institute{Institut f{\"u}r
    Informatik\\Ludwig-Maximilians-Universit{\"a}t M{\"u}nchen\\
    Munich, Germany}
\email{\{wirsing,bauerse,schroeda\}@pst.ifi.lmu.de} }
\def\titlerunning{Modeling and Analyzing Adaptive User-Centric Systems
  in Real-Time Maude} \def\authorrunning{M. Wirsing, S.~S. Bauer \&
  A. Schroeder}
  
\begin{document}
\maketitle

\begin{abstract}
Pervasive user-centric applications are systems which are meant to sense the presence, mood, and intentions of users in order to optimize user comfort and performance or to assist people in their specific activities. 
Building such applications requires not only 
state-of-the art techniques from artificial intelligence but also sound software engineering methods for facilitating modular design, runtime adaptation and verification of critical system requirements.  

In this paper we focus on high-level design and analysis, and use the algebraic rewriting language Real-Time Maude  for specifying applications in a real-time setting. We propose a component-based approach for modeling pervasive user-centric systems in a generic way and show how to instantiate the generic rules for a simple
out-of-home digital advertising application and how to analyze and prove crucial properties of the system architecture through model checking and simulation. For proving time-dependent properties of systems we use 
Metric Temporal Logic (MTL) and present analysis algorithms for model checking two subclasses of MTL formulas: time-bounded response and time-bounded safety MTL formulas.   The underlying idea is to extend the Real-Time Maude model with suitable clocks and to transform the MTL formulas into  LTL formulas over the extended specification. This makes it possible to use the LTL model checker of Maude for verifying real time system properties. 
It is shown that component-based Real-Time Maude specifications as well as their extensions by clocks are time-robust and finite state; moreover, the above classes of formulas are tick-stabilizing if their atomic propositions are tick-stabilizing. As a consequence, model checking analyses are sound and complete for maximal time sampling.

The approach is illustrated by a simple adaptive advertising scenario in which an adaptive advertisement display can react to actions of the users in front of the display.

\emph{Keywords:} Component-based software engineering, reconfiguration, algebraic specification, term rewriting, Real-Time Maude, real-time temporal logic
\end{abstract}

\section{Introduction}\label{sec:introduction}

As we are moving on from desktop computers to a pervasive computing
intelligence interwoven in the ``fabric of everyday life'', our
environment is about to become enriched with more and more smart
assistance systems.

Through this transformation, it becomes feasible for IT systems in our
environment to measure responses of
the user's body through sensors and cameras and to influence our physical, emotional and cognitive state for our, the users, benefits. We call such systems pervasive user-centric applications \cite{satyanarayananm2001}.
Examples are a so called ``mood player'' which selects the music according to current mood of a person, a ``driving assistant'' which implements adaptive control in vehicles to achieve more secure, more
pleasant and more effective driving, or ``adaptive advertising'' where the displayed content of an advertisement is dynamically adapted to the needs of the actual audience in front of the display~\cite{reflect-project}.

%On the technical side, pervasive user-centric applications constitute a biocybernetic
%loop~\cite{serbedzija09biocybernetic} in which the reactions of the
%user alter the behaviour of the software system and the actions of the
%system influence the user, thereby changing his reactions.  Through
%continuous interaction, such a biocybernetic loop is able to get to
%know the user and adapt the environment to his needs. If designed
%carefully, the IT-system can interact through implicit channels with
%the user and create a more productive, more comfortable, or more
%suitable environment.

Building pervasive user-centric applications is not easy, and requires
state-of-the art techniques from artificial intelligence including
machine learning and probabilistic reasoning, as well as a lot of system calibration and experimental psychological research  in order to determine the right sensor parameters for recognizing the mood or the cognitive state of a person. As a consequence, from the software engineering point of view it is important that such systems are easily changeable and adaptable at runtime; moreover, they need to react immediately to the behavior of the user and thus have to satisfy (soft) real-time constraints. In the REFLECT project~\cite{reflect-project} we have developed a  component-based framework ~\cite{ucpa09:beyer-hammer-kroiss-schroeder} which facilitates modular design, runtime adaptation and reconfiguration of systems and supports the implementation of pervasive user-centric applications (such as the ones mentioned above) in a flexible way. 

In this paper we focus on the high-level design and analysis of pervasive user-centric applications in order to be able to make guarantees on the correct behavior of such systems in an early stage of development. We follow the algebraic paradigm based on term rewriting and use Real-Time Maude as a high-level formal modeling language for pervasive user-centric applications in a real-time setting. 

In line with the REFLECT framework we propose a component-based approach for modeling pervasive user-centric applications in a generic way and show how to instantiate the generic rules for a simple
out-of-home digital advertising application and how to analyze and prove crucial properties of the system architecture through model checking and simulation.   

In our approach components are considered to be black boxes, making explicit only
their communication requirements by means of \emph{required} and
\emph{provided} ports. A system \emph{configuration} comprises a number of \emph{components} and
 \emph{connectors}, which describe how the required ports are
connected to suitable provided ports. 
%Numerous component models
%describe how exactly this is achieved, and how the components can
%communicate in order to achieve a common goal~\cite{LW07SCM}.
We distinguish three kinds of components: basic components, timed components whose behavior is influenced by timers, and hierarchical components (often just called components)  which typically contain other components and connectors as well as timers. Generic rules are defined for transmitting values along connectors as well as for time elapse and the specifications of different kinds of timers.

Individual components provide parts of the
functionality required by the entire system. By changing connectors, adding
and removing individual components, the system's behavior can be
changed at runtime. This process is called \emph{dynamic reconfiguration}. Since entire
components are replaced, little code needs to be added to the
components to achieve this kind of adaptivity. Instead, it is attained
on the level of the system architecture. In our approach, timed
\emph{monitor} components survey the behavior of the 
system and the environment, and trigger reconfigurations if necessary.

%Using a component-based approach supporting reconfiguration allows to
%use an additional structuring mechanism: not only to structure an
%application as a set of collaborating components, but to dynamically
%replace components to create configurations fitting for a subset of
%the overall intended scenarios, and changing the configuration between
%these partial solutions.  In our approach, a system contains
%\emph{monitors} which are dedicated components for surveying the
%system and the environment. Monitors allow to deploy system
%configurations that do not exhibit the desired behaviour only under
%specific assumptions, as they trigger reconfigurations as these
%assumptions do not hold any more.

For proving properties of systems we use 
Metric Temporal Logic (MTL) \cite{DBLP:journals/rts/Koymans90}. This is an 
extension of  Linear Temporal Logic (LTL) \cite{kroegerTL} for specifying
 \emph{timed} properties.  Currently, 
Real-Time Maude does not provide an MTL model
checker. However, in previous works, cf.,
e.g.~\cite{DBLP:conf/snpd/Olveczky08}, {\"O}lveczky showed how to verify some simple MTL formulas by 
using the time-bounded search command of Real-Time Maude or the LTL
model checker of Maude. In
\cite{rtrts-Olveczky}, Lepri et al.\ present an automatized analysis algorithm of
two important classes of MTL formulas, namely the bounded response
property $\always(p \to (\eventually_{\leq b} q))$ and the minimum
separation property $\always(p \to (p\weakuntil (\always_{\leq b} \neg
p)))$. The underlying idea is to extend a Real-Time Maude model by a suitable clock and to transform the MTL formulas into  LTL formulas over the extended specification.  Then the LTL
model checker of Maude can be used for performing the analysis.

In this paper, we extend these ideas and present analysis algorithms for two further
and more general classes of MTL formulas:
\begin{enumerate}
\item Generalized time-bounded response: $\always(\bigvee_{i\in I}(\eventually_{\leq b_i} q_i))$ for $I = \{1,2,\ldots,n\}\subset \mathbb N$ a finite set of indices, and
\item Time-Bounded safety: $\always(p \lor \always_{\leq b} q)$
\end{enumerate}
(where $q_i$, $q$, and $p$ are all atomic propositions).

We show that component-based Real-Time Maude specifications as well as their extensions by clocks are time-robust and finite state; moreover, the above classes of formulas are tick-stabilizing if their atomic propositions are tick-stabilizing. As a consequence (cf.~\cite{DBLP:journals/entcs/OlveczkyM07a}, model checking analyses are sound and complete for maximal time sampling.

Throughout the paper we  illustrate our modeling and analysis techniques by a simple scenario of adaptive advertisement.
 
The paper is organized as follows: In Section~\ref{sec:advertising} we
present the adaptive advertising case study which we use as running
example. The following Section~\ref{sec:rtmaude} contains a short
introduction to Real-Time Maude. Sections~\ref{sec:modeling}
and~\ref{sec:analysis} present our main results. In
Section~\ref{sec:modeling} we explain our generic framework for
specifying component-based systems and the Real-Time Maude
specification of the adaptive advertisement scenario. The
transformation algorithms for the time-bounded response and
time-bounded safety formulas are presented in
Section~\ref{sec:analysis}; we show also completeness and termination
of LTL model checking for our format of component-based specifications
and illustrate our results by applying the Maude model checker
successfully to the requirements of the adaptive
advertisement scenario. In Sections~\ref{sec:related}
and~\ref{sec:concluding} we discuss related work, summarize our
results and discuss further work.

\section{An adaptive advertising application}\label{sec:advertising}

To showcase our approach to system verification, we consider a simple
out-of-home digital advertising application~\cite{advertising09:beyer-mayer-kroiss-schroeder}. The setup of this
application consists of a large display screen and a camera monitoring
the area in front of the display, and by this allowing interactions of
passer-bys with the displayed content. The general idea of adaptive
advertising is to adapt a displayed advertisement to the current
situation in front of it -- whether there are several people just
passing by, a small group of persons watching the ad carefully, or
just one person in front of it waiting for someone else
\cite{advertising09:beyer-mayer-kroiss-schroeder}. In this simplified example, we consider the
camera as a way to enable gesture-based interactions with a passer-bys
and to discover their presence.

A simple scenario within the adaptive advertising setting is an adaptive car
advertisement, reacting to gestures of users in front of the display: By
moving around the display, pointing at items or looking at them, the users
influence the contents of the ad.
To function properly this system should satisfy the following two requirements: 
(G1) Being an interactive ad, the system should react to a
user in front of the display. (G2) The content displayed must change
at least every ten seconds: an advertising campaign using a
large-scale display should not waste its capabilities by showing
static content.

\begin{figure}[t]%
  \resizebox{\columnwidth}{!}{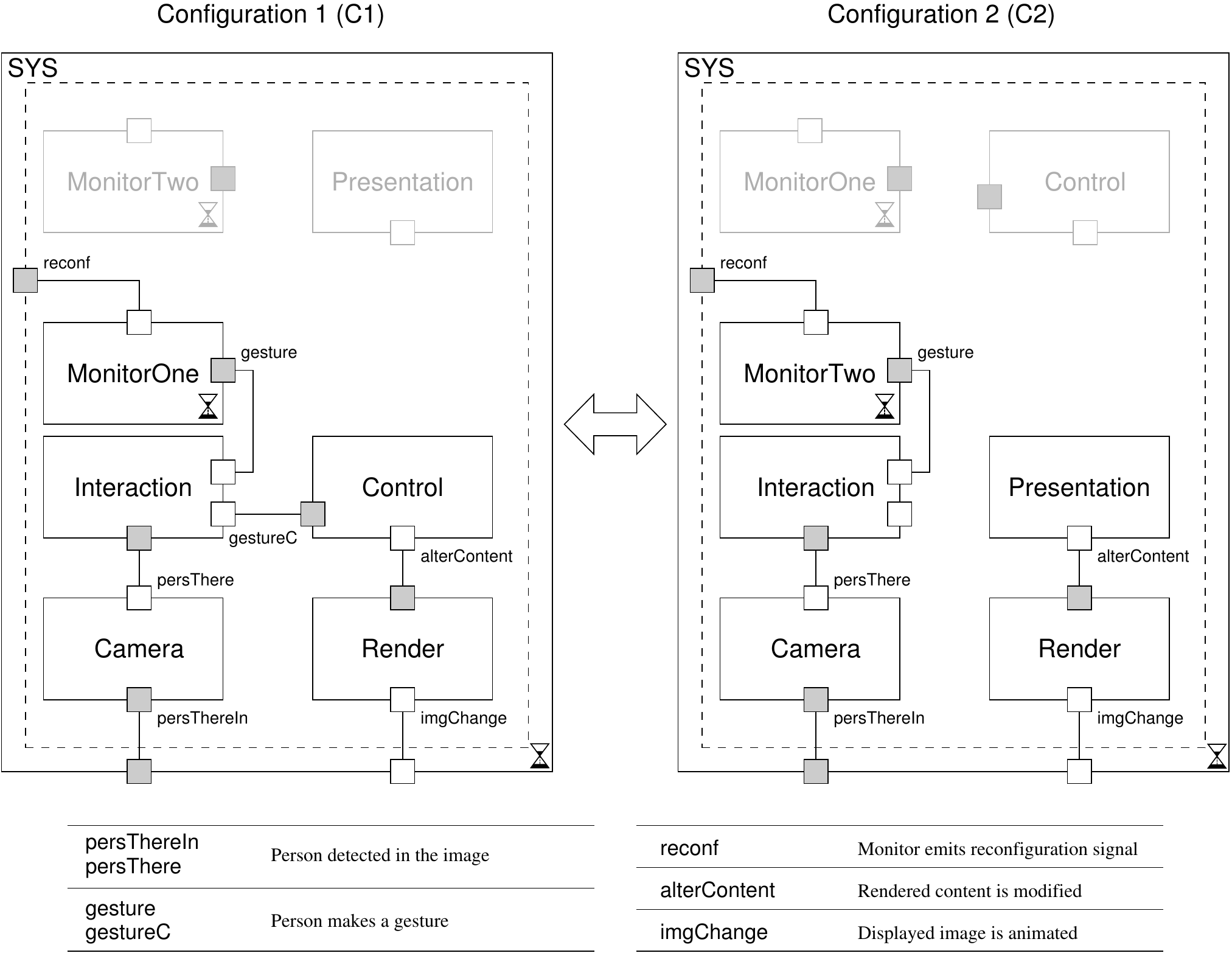 }%
  \caption{Adaptive advertising system configurations}%
  \label{fig:reconfig-sys}%
\end{figure}

The realization shown in figure~\ref{fig:reconfig-sys} (components
that are present but inactive are shown in light-gray) is first
deployed in interactive mode and monitors whether someone is
interacting with the ad. If that is not the case, a reconfiguration is
triggered which altering the system configuration so that it shows
auto-active content generated by a presentation component, e.g.~an
advertising movie or predefined animation sequences. Introducing
monitors allows a partial solution to assume that the environment
exhibits certain features (e.g.~always have someone interacting with
the ad) that it does not exhibit in the general case. Note that the
second system (figure~\ref{fig:reconfig-sys}, right) also needs
monitoring, as it again does not satisfy (G1): The second system
provides interactive content to its viewers, and therefore must be
changed as soon as a person is in front of the display, interacting
with it. 

Reconfiguration leads to further requirements; in particular, system configurations should reasonably stable so that the system does not oscillate between several configurations. This can be expressed as follows: (G3) a reconfiguration should not happen instantaneously, but
must take at least 200 ms to complete.

\section{Real-Time Maude}\label{sec:rtmaude}

Real-Time Maude \cite{DBLP:journals/lisp/OlveczkyM07} is a formal
specification language based on Maude
\cite{DBLP:conf/maude/2007}, a high-performance simulation and model
checking tool which uses rewriting logic and membership equational
logic for the specification of systems. Real-Time Maude extends Maude
by supporting the formal specification of \emph{real-time} system
while benefiting from the expressiveness of the Maude language and
the powerful analysis techniques like LTL model checking. In this
section, we will briefly introduce the main concepts of specifications
in Real-Time Maude; we refer to \cite{DBLP:journals/lisp/OlveczkyM07}
for more details on the syntax and semantics of Real-Time Maude.

In Real-Time Maude, real-time systems are formally specified by a
real-time rewrite theory of the form $\mathcal R = (\Sigma, E, IR,
TR)$ where $(\Sigma, E)$ is a membership equational logic
\cite{DBLP:conf/maude/2007} theory with $\Sigma$ a signature and $E$ a
set of confluent and terminating conditional equations, and $IR$ is a
set of instantaneous (rewrite) rules specifying the system's
transitions which happen in zero time. Instantaneous rules are written
\[\small \texttt{crl [$l$] : $t$ => $t'$ if $cond$ .}\] where $l$ is a label, 
$t$, $t'$ are terms, and $cond$ is a condition on the terms $t$, $t'$.
Finally, $TR$ is a set of tick (rewrite) rules which specify how the
system behaves when time advances. Tick rules are written \[\small \texttt{crl
  [$l$] : \{$t$\} => \{$t'$\} in time $T$ if $cond$ .}\] where
\texttt{\{\_\}} is a constructor of sort \verb|GlobalSystem|, and $T$
is a term of sort \verb|Time| which denotes the duration of the tick
rule. The form of the tick rules ensure that time advances uniformly
in the whole system.

A one-step rewrite, written $t \xrightarrow{r} t'$, is a single
rewrite of a term $t$ to a term $t'$ (both of sort
\verb|GlobalSystem|) in time $r$ (possibly zero time). We call $t$ the source
state of the rule $t \xrightarrow{r} t'$, and $t'$ the target state.
A (timed) path in $\mathcal R$ is an infinite sequence $\pi = t_0
\xrightarrow{r_0} t_1 \xrightarrow{r_1} t_2 \ldots$ where either for
all $i \in \mathbb N$, $t_i \xrightarrow{r_i} t_{i+1}$ is a one-step
rewrite of $\mathcal R$, or there exists a $k \in \mathbb N$ such that
for all $0 \leq i < k$, $t_i \xrightarrow{r_i} t_{i+1}$ is a one-step
rewrite in $\mathcal R$ and there is no one-step rewrite from $t_k$ in
$\mathcal R$, and $t_j = t_k$ and $r_{j-1}=0$ for each $j > k$. The
set of all timed paths of $\mathcal R$ starting in $t$ is denoted by
$\Paths(\mathcal R)_{t}$. For $k \in \mathbb N$, we define $\pi^k$ to
be the timed path starting after the $k$th one-step rewrite,
i.e.~$\pi^k=t_k \xrightarrow{r_k} t_{k+1} \xrightarrow{r_{k+1}}
t_{k+2}\ldots$.  A term $t'$ is reachable from a term $t$ in $\mathcal
R$ in time $r$ if there is a path $\pi = t_0 \xrightarrow{r_0} \ldots
\xrightarrow{r_{k-1}} t_k \xrightarrow{r_{k}} \ldots$ such that $t_k =
t'$ and $\sum_{i=0}^{k-1}r_i$.

Function symbols $f$ are declared by the statement \texttt{$f$\ :\
  $s_1 \ldots s_n$ -> $s$} with sorts $s_1,\ldots,s_n,s$, and
equations are written \texttt{eq $t$ = $t'$}. A variable $x$ of sort
$s$ is declared by the statement \texttt{var $x$\ :\ $s$}.

In \emph{object-oriented} Real-Time Maude, classes are declared by
\[\small \texttt{class\ $C$\ |\ $att_1$\ :\ $s_1$,\ $\ldots$,\ $att_n$\ :\
  $s_n$\ .}\]
where $att_1, \ldots, att_n$ are attributes of sorts $s_1, \ldots,
s_n$, respectively. An object of class $C$ is written as a term
\[\small \texttt{<\ $o$\ :\ $C$\ |\ $att_1$\ :\ $val_1$,\ $\ldots$,\ $att_n$\
  :\ $val_n$\ >}\] of
sort \texttt{Object}, where $o$ is an object identifier of sort
\texttt{Oid} and $val_i$ are the current values of attributes $att_i$
($1 \leq i \leq n$). A system state is a collection of
objects\footnote{In Real-Time Maude, states in object-oriented
  specifications usually contain messages which are used to model
  communication between objects. However, in our case study, we will
  not make use of messages and let objects ``directly'' communicate,
  i.e.~the effect of communication between two objects is modeled by a
  rewrite rule having both objects in source and target state.} and is
of sort \verb|Collection| which is a multiset equipped with an associative and
commutative union operator with empty syntax, e.g.
\[\small \texttt{<\ $o$\ :\ $C$\ |\ $att_1$\ :\ $val_1$,\ $\ldots$,\ $att_n$\ :\ $val_n$\ >
  <\ $o'$\ :\ $C'$\ |\ $att_1'$\ :\ $val_1'$,\ $\ldots$,\ $att_m'$\ :\ $val_m'$\ > $\ldots$}\]
represents a system state consisting of the objects $o$, $o'$,
... . Real-Time Maude supports \emph{multiset rewriting}, i.e.~rewrite rules
are applied modulo associative and commutative rewriting of the system
state.  In object-oriented Real-Time Maude specifications, the
time-dependent behavior is usually specified by a single \emph{tick rule}
of the form\\[0.3cm]
\small
\verb|var C : Configuration . var T : Time .|\\
\verb|crl [tick] : {C} => {delta(C,T)} in time T if T <= mte(C) /\ | $cond$ \verb| [nonexec] .|\\[0.3cm]
\normalsize 
The function \verb|delta| defines the effect of time
  elapse on a configuration, and the function \verb|mte| defines the
  maximum amount of time that can elapse before some action must take
  place. These functions distribute over the objects in a
  configuration and must be defined for all single objects to defined
  the timed behavior of a system. The tick rule advances time
  nondeterministically by any amount \verb|T| less than or equal to
  \verb|mte(C)|. To execute such rules, Real-Time Maude offers a
  number of time sampling strategies, so that only some moments in
  time are visited. In this paper, we will only make use of the
  maximal time sampling strategy which advances time to the next
  moment when some action must be taken, as defined by \verb|mte|,
  i.e.~when the tick rule is applied in a state $\{\tt C\}$, time is
  advanced by \verb|mte(C)|.
  The above form of the tick rule slightly
  differs from the usual form proposed in
  \cite{DBLP:journals/lisp/OlveczkyM07} by allowing an additional
  condition $cond$ (\verb|T| must not occur in $cond$).

  A Real-Time Maude specification is executable and various formal
  analysis methods are supported. For a complete overview of these
  methods see, e.g.,~\cite{DBLP:journals/lisp/OlveczkyM07}. In this work we
  make use of the \emph{time-unbounded} model checking command
  \[\texttt{(mc $t$ |=u $\phi$ .)}\]
  for an initial state $t$ and a temporal logic formula $\phi$.

  In the rest of this paper, when we talk of a real-time rewrite
  theory $\mathcal R$ we typically mean the real-time rewrite theory
  $\mathcal R^{max}$ which is obtained from $\mathcal R$ by applying
  the theory transformation corresponding to using the maximal time
  sampling strategy when executing the tick rules.
%\footnote{eigentlich
%    $\mathcal R^{max,nz}$ ?}
%TODO: and non-zeno paths???? see completenss result

\section{Modeling in Real-Time Maude}\label{sec:modeling}
In this section we show how our case study, the digital 
advertising application, as
described in Sect.~\ref{sec:advertising},
can be modeled in Real-Time Maude. For this purpose, we first
present in Sect.~\ref{sec:components} an implementation of a generic,
port-based component model in object-oriented Real-Time Maude. Then we
show in Sect.~\ref{sec:modeling_example} how our case study described
in Sect.~\ref{sec:advertising} can be modeled as self-reconfiguring
component-based system.

% non-flat!

\subsection{Defining Components in Object-Oriented Real-Time
  Maude}\label{sec:components}

Components are encapsulated entities with explicit ports over which
communication take place. A port is modeled as an object instance of
the class \verb|Port| having one attribute \verb|value|
describing the current state of the port.
\begin{alltt}\small
  class Port | value : Bool .
\end{alltt}
The value of a port models the state of activity: a value \verb|true|
models the fact that at this port, a signal is received (or sent) via
this port whereas \verb|false| means that no signal is received (or
sent).

A port always belongs to a unique component and may have two roles:
either it is a \emph{provided} port or a \emph{required} port of this
component. All provided ports are under the control of the owning
component and therefore, the state (i.e.~the value) of a provided port
can only be changed by its owning component. In contrast, the value of
a required port cannot be changed by the owning component but can only
be changed by the environment -- in this sense, a component can only
react on different states if its required ports.

We introduce three different types of components: \emph{basic
  components}, \emph{timed components}, and \emph{hierarchical
  components} (or just \emph{component}). For these three different
types of components we introduce the following (sub-)classes using
inheritance, i.e.~timed component inherits from basic component, and
component inherits from timed component.
\begin{figure}[t]
  \centering {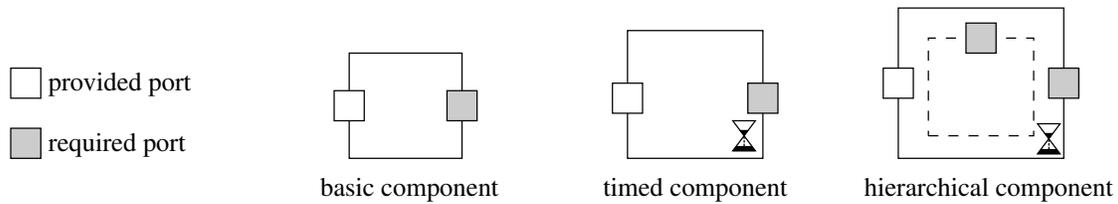 }
  \caption{Ports, Components}
  \label{fig:components}
\end{figure}
\begin{alltt}\small
  class ABasicComponent | prov     : Configuration, req      : Configuration .
  class ATimedComponent | tstate   : Configuration .
  class AComponent      | assembly : Configuration, innerreq : Configuration .
  subclass ATimedComponent < ABasicComponent .
  subclass AComponent      < ATimedComponent .
\end{alltt}
% A basic component has a local data state (\verb|cstate| of type
% \verb|CState|) which is used to store local data in a component
% which is invariant with respect to time progress, i.e.~\verb|cstate|
% does not change when time progresses.
The different types of components are also illustrated in
Fig.~\ref{fig:components}.  A basic component has a set of provided
(\verb|prov|) and required ports (\verb|req|), and both attributes are
modeled as instances of type \verb|Configuration|. However, we assume
that the multisets \verb|prov| and \verb|req| only contain object
instances of type \verb|Port|. A timed component inherits from basic
component and has an additional attribute (\verb|tstate|) which models
a timed data state which need \emph{not} be time invariant. For timed
data states we refer the reader to the end of this section. Finally,
hierarchical components embody an inner assembly (\verb|assembly|) of
connectors and components, and inner required ports (\verb|innerreq|)
that are connected to provided ports of components within the inner
assembly.

For each type of component we actually introduce two classes, e.g.~for
basic components, we define an abstract class \verb|ABasicComponent|
and a concrete class \verb|BasicComponent| and
only allow object instances of the concrete class. This discrimination between abstract 
and concrete classes allows to define both common rewrite rules and equations applicable 
to all component types, and rules and equations applicable to a particular type of components 
only. The class definitions are hence as follows:
% The most primitive components are basic components
% (\verb|ABasicComponent|) which have a data state (\verb|cstate|) and
% provided (\verb|prov|) and required (\verb|prov|) ports. Timed
% components have an additional timed state (\verb|tstate|) which is
% not necessarily invariant to time progress, and (hierarchical)
% components have an internal configuration \verb|assembly| and inner
% required ports (to internal components in the assembly,
% \verb|innerreq|).
\begin{alltt}\small
  class BasicComponent .
  class TimedComponent .
  class Component .
  subclass BasicComponent < ABasicComponent .
  subclass TimedComponent < ATimedComponent .
  subclass Component      < AComponent .
\end{alltt}

For timed components (and hence for hierarchical components), a timed
data state (attribute \verb|tstate|) is a set of timers which
decrement their value by the advanced time. In our case study later
on, we need three different types of timers, all modeled as classes.
\begin{alltt}\small
  class Timer      | value : TimeInf .
  class OnOffTimer | value : TimeInf, active : Bool .
  class DelayTimer | value : TimeInf, delay  : TimeInf .
\end{alltt}
The class \verb|Timer| models a simple timer which has an attribute
\verb|value| for the current time value of type \verb|TimeInf|. The
class \verb|OnOffTimer| has an additional attribute \verb|active| to
switch the timer on and off. Finally, the third class
\verb|DelayTimer| is another timer class which contains -- beside the
timer value -- an attribute \verb|delay|. If the \verb|DelayTimer|
expires the timer value is reset to the fixed delay.

Components communicate over their required and provided ports which
are connected by \emph{connectors}. More precisely, we distinguish
between two types of connectors: On the one hand, a class
\verb|Connector| models all connectors which link a provided port with
a required port on the same level of the component hierarchy (i.e.~not
crossing component boundaries). On the other hand, a class
\verb|DelegateConnector| models all delegate connectors which link,
within hierarchical components, ports in the assembly with either
outer ports or inner required ports.

An overview of all types of connectors is given in
Fig.~\ref{fig:connectors}.
\begin{figure}[t]
  \centering 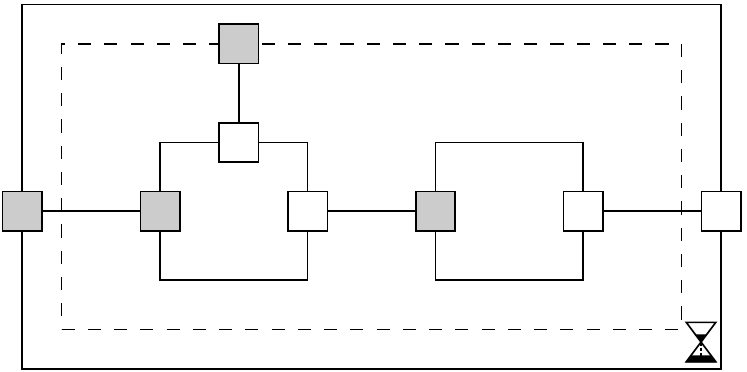
  \caption{Connectors and Delegate Connectors}
  \label{fig:connectors}
\end{figure}
\begin{enumerate}[(1)]
\item A connector links a provided and required port, on the same
  level of component hierarchy.
\item A delegate connector links either
  \begin{enumerate}[(a)]
  \item a provided port of an inner component (source) with an outer
    provided port (target), or
  \item a provided port of an inner component (source) with an inner
    required port of the comprising hierarchical component (target),
    or
  \item an outer required port (source) with a required port of an
    inner component (target).
  \end{enumerate}
\end{enumerate}

Connectors are again modeled as object-oriented classes
\verb|Connector| and \verb|DelegateConnector|, each having two
attributes \verb|source| and \verb|target| which are object
identifiers of instances of class \verb|Port|.
\begin{alltt}\small
  class Connector | source : Oid, target : Oid .
  class DelegateConnector | source : Oid, target : Oid .
\end{alltt}

% DIRECTION!!  with the meaning that the port with object identifier
% \verb|source| (which is a provided port) is connected to the port
% with object identifier \verb|target| (which is a required
% port). Connectors are required not to cross boundaries of
% components, so only ports on the same abstraction level may be
% connected. For the connection of outer ports (of hierarchical
% components) with ports of inner components, we define a class
% \verb|DelegateConnector| which connects

Component behavior is modeled in an abstract way by defining an
operation \verb|beh| on configurations.
\begin{alltt}\small
  op beh : Configuration -> Configuration [frozen (1)] .
\end{alltt}
The behavior of a component can be defined then by introducing
equations for \verb|beh|. This abstract behavior operation is used in
the following generic (instantaneous) rewrite rules.

We define generic, instantaneous rules for transmitting values along
connectors. The following rule \verb|[transmit]| assumes two
(arbitrary) components with connected provided and required ports; if
the value of the provided port is not equal to the required port
value, then the latter value is changed accordingly such that
afterwards, both connected ports have equal values. A necessary condition
for this rule is that the target component which alters one of its ports
is in a consistent state. A (hierarchical) component is called consistent 
if within its assembly, all connected ports are equal in value.

\begin{alltt}\small
crl [transmit] :
	   < o1 : ABasicComponent | prov : < p1 : Port | value : b > PORTS >
	   < c : Connector | source : p1, target : p2 >
	   < o2 : ABasicComponent | req : < p2 : Port | value : b' > PORTS' >
	  =>
	   < o1 : ABasicComponent | prov : < p1 : Port | value : b > PORTS >
	   < c : Connector | source : p1, target : p2 >
	   beh(< o2 : ABasicComponent | req : < p2 : Port | value : b > PORTS' >)
	if b =/= b' and
	   consistent(< o2 : ABasicComponent | req : < p2 : Port | value : b' > PORTS' >) .
\end{alltt}

Note that in the above rule, the abstract operation \verb|beh|
modeling the component behavior is called on the receiving component
which may react on its new state, more precisely, the altered state of
its required ports. It is also worth mentioning that this rule is
defined uniformly on all types of components by using
\verb|ABasicComponent|; every ``concrete'' component (of type
\verb|BasicComponent|, \verb|TimedComponent|, \verb|Component|) is a
subclass and hence inherits this rewrite rule.

For propagation of port values along connectors, each type of
connector has its own behavior modeled as a rule. The rule
\verb|[transmit]| models the behavior of a connector, and since we
have three different types of delegate connectors, we have also three
more rules which do not described here in detail: A rule
\verb|[delegateIn]| equals inner required port with outer required
port, a rule \verb|[delegateOut]| equals outer provided port with
inner provided port, and a rule \verb|[delegateInnerPort]|
equals an inner required port of the comprising hierarchical component
with an inner required port (in the assembly).

Finally, we have the tick rule which advances time up to the maximal
possible amount of time determined by the function \verb|mte|.
{\small \begin{verbatim}
var C : Configuration .   var T : Time .
crl [tick] : {C} => {delta(C,T)} in time T if T <= mte(C) /\ consistent(C) [nonexec] .
\end{verbatim}}

It is important to point out that we only let time advance if the
system is in a consistent state, i.e.~the term \verb|consistent(C)|
evaluates to \verb|true| if and only if all connected ports are equal
in value. The functions \verb|delta|, which models the effect of time
elapse on the system state, and \verb|mte|, which for a system state
returns the maximal possible time elapse, are defined as usual in
object-oriented specifications in Real-Time Maude
(cf.~\cite{DBLP:journals/lisp/OlveczkyM07}), e.g.~for the class
\verb|OnOffTimer|, we have the following equations:

{\small \begin{verbatim}
eq delta(< o : OnOffTimer | value : t, active : b >, T)
 = < o : OnOffTimer | value : if b then t monus T else t fi, active : b > .

eq mte(< o : OnOffTimer | value : t, active : b >)
 = if b then t else INF fi .
\end{verbatim}}

\subsection{Modeling the Digital Advertising Application}\label{sec:modeling_example}
We show now how the adaptive advertising application of
Sect.~\ref{sec:advertising} can be modeled as a component-based system
in Real-Time Maude by extending the implementation introduced so far.

The system, cf.\ Fig.~\ref{fig:reconfig-sys}, is modeled as a
hierarchical component with object identifier \verb|SYS|. It has one
outer provided port \verb|SYS.imgChange| and one outer required port
\verb|SYS.persThereIn|. For receiving a reconfiguration signal from
the inner assembly the system component contains an inner required
port \verb|SYS.reconf|. The (timed) state consists of a timer
\verb|reconftimer| which is used to trigger reconfiguration of the
component after a (fixed) time delay. The inner assembly comprises
delegate connectors (e.g.\ \verb|d1| connecting \verb|SYS.persThereIn|
and \verb|Camera.persThereIn|), connectors (e.g.\ \verb|c1| connecting
\verb|Camera.persThere| and \verb|Interaction.persThere|), basic
components (e.g.\ \verb|Render|), and timed components
(\verb|MonitorOne| and \verb|MonitorTwo|). Each monitor is equipped
with a timer, more precisely, an object instance of the class
\verb|OnOffTimer|. The purpose of these timers is to count the time
when the condition for reconfiguration to the other configuration is
true.
\begin{alltt}\small
op SYS_in_C1 : -> Configuration [ctor] .
eq SYS_in_C1 = < SYS : Component |
                   prov : < SYS.imgChange : Port | value : true >,
                   req : < SYS.persThereIn : Port | value : true >,
                   tstate : < reconftimer : Timer | value : INF >,
                   innerreq : < SYS.reconf : Port | value : false >,
                   assembly :
                     < d1 : DelegateConnector | source : SYS.persThereIn,
                                                target : Camera.persThereIn >
                     ...
                     < c1 : Connector | source : Camera.persThere,
                                        target : Interaction.persThere >
                     ...
                     < Render : BasicComponent |
                        prov : < Render.imgChange : Port | value : true >,
                        req  : < Render.alterContent : Port | value : true > >
                     < Presentation : BasicComponent | 
                        prov : < Presentation.alterContent : Port | value : true >,
                        req  : none >
                     ...
                     < MonitorOne : TimedComponent |
                        tstate : < m1timer : OnOffTimer | value : INF, active : true >,
                        prov : < MonitorOne.reconf : Port | value : false >,
                        req  : < MonitorOne.gesture : Port | value : true > >
                     ... > .
\end{alltt}
The component behavior is defined via equations for the operation
\verb|beh|. Note that \verb|beh| applied to components which have
altered the port values of required ports and have to react
accordingly (e.g.\ \verb|beh| is called in the rule \verb|[transmit]|,
cf.\ Sect.~\ref{sec:components}). For instance, for component
\verb|Render| we define \verb|beh| by propagating the new value of the
required port (which has been previously changed by rule
\verb|[transmit]|) to the provided port.\footnote{In Real-Time Maude,
  object instances in terms need not list all attributes; it is valid
  to omit attributes which are not relevant.}
\begin{alltt}\small
eq beh(< Render : BasicComponent |
            prov : < Render.imgChange    : Port | >,
            req  : < Render.alterContent : Port | value : b' > >)
     = < Render : BasicComponent |
            prov : < Render.imgChange    : Port | value : b' >,
            req  : < Render.alterContent : Port | > > .
\end{alltt}
The behavior of the monitor \verb|MonitorOne| can be defined
similarly, with a more involved behavior of the
timer. \verb|MonitorOne| is active in configuration C1 and surveys
whether there is a person in front of the display. If there is no
person in front (\verb|MonitorOne.gesture| becomes false), the timer
is initialised with 2000. Otherwise, if the valuation of port
\verb|MonitorOne.gesture| has changed from false to true and the timer
has not expired already, the timer is set to infinite (\verb|INF|), i.e.\
if time advances (by a finite amount of time), the timer is not decreased.
\begin{alltt}\small
eq beh(< MonitorOne : TimedComponent |
           tstate : < m1timer : OnOffTimer | value : t, active : B >,
           req    : < MonitorOne.gesture : Port | value : b > >)
     = < MonitorOne : TimedComponent |
           tstate : < m1timer : OnOffTimer | value : if (b or (not B)) and t =/= 0
                                                      then INF else (if t == INF
                                                        then 2000 else t fi) fi > .
\end{alltt}
Note that the timer \verb|m1timer| of \verb|MonitorOne| is not touched
if the timer has expired, i.e.\ has 0 as value. In this case, for a
period of 2000 ms, the port valuation of
\verb|MonitorOne.gesture| has been false. To guarantee the overall
system guarantee (G2) which says that the display's content should
change at least every ten seconds, the system must be
reconfigured. Reconfiguration is signaled to the component \verb|SYS|
by setting the value of the port \verb|MonitorOne.reconf| to
\verb|true|.
\begin{alltt}\small
rl [monitorOne-signal] :
  < MonitorOne : TimedComponent |
      tstate : < m1timer : OnOffTimer | value : 0, active : true >,
      prov   : < MonitorOne.reconf : Port | value : false > >
 =>
  < MonitorOne : TimedComponent |
      tstate : < m1timer : OnOffTimer | value : INF, active : false >,
      prov   : < MonitorOne.reconf : Port | value : true > > .
\end{alltt}
This value of the reconfiguration port of the monitor is propagated
to the component \verb|SYS| which then sets its reconfiguration timer to
250 ms which models the duration of the reconfiguration process.
\begin{alltt}
eq beh(< SYS : Component |
          tstate : < reconftimer : Timer | value : t >,
          innerreq : < SYS.reconf : Port | value : true > >)
     = < SYS : Component |
          tstate : < reconftimer : Timer |
                       value : if t == INF then 250 else t fi > > .
\end{alltt}
The reconfiguration is performed as soon as the timer expires:
The reconfiguration timer is set to \verb|INF|,
the connectors of configuration 1 are replaced by the connectors
configuration 2, and, moreover, the second monitor is activated which
must observe the component to trigger a reconfiguration back to configuration 1
if the assumptions of configuration 2 are not met any more.
\begin{alltt}\small
rl [reconf-C1-to-C2] :
  < SYS : Component | tstate : < reconftimer : Timer | value : 0 >,
     assembly : ... {\it connectors of configuration 1\/} ...
       < MonitorTwo : TimedComponent |
          tstate : < m2timer : OnOffTimer | value : INF, active : false >, ... >
       ... > 
 =>
  < SYS : Component | tstate : < reconftimer : Timer | value : INF >,
     assembly : ... {\it connectors of configuration 2\/} ...
       < MonitorTwo : TimedComponent |
          tstate : < m2timer : OnOffTimer | value : INF, active : true >, ... >
	    ... > .
\end{alltt}
We omit the rest of the rewrite rules for the remaining components and
their equations for \verb|beh| defining their behavior.  They are in
fact all very similar to the rules presented above. It is, however,
worth mentioning that the timer of \verb|MonitorTwo| is set to 500
ms as soon as there is a person detected in front of the
display.

So far we have defined the term \verb|SYS_in_C1| of sort
\verb|Configuration| and introduced appropriate rewrite rules and
equations which model the behavior of the hierarchical
component. However, it is an open component with a required port
\verb|SYS.persThereIn|; the overall behavior of component \verb|SYS|
depends on how the valuation of \verb|SYS.persThereIn| evolve over
time. To allow analysis of the system, we add a component \verb|ENV|
which models the environment; it has two ports which are connecting to
their counterparts of the component \verb|SYS|, thus yielding a closed
system. The initial state is then defined as:
\begin{alltt}\small
  op initial : -> Configuration [ctor] .
  eq initial =
    SYS_in_C1
    < ENV : TimedComponent | prov   : < ENV.persThereIn : Port | value : true >,
                             req    : < ENV.imgChange : Port | value : true >,
                             tstate : < envdtimer : DelayTimer | value : 0,
                                                                 delay : 50 > >
    < CONN1 : Connector | source : ENV.persThereIn, target : SYS.persThereIn >
    < CONN2 : Connector | source : SYS.imgChange, target : ENV.imgChange > .
\end{alltt}
The behavior of the environment is as follows:
Every 50 ms the environment non-deterministically choose whether
to change the valuation of port \verb|ENV.persThereIn|, or not.
This recurring choice after 50 ms is modeled by a delay timer, always resetting the timer after each choice.
\begin{alltt}\small
rl [env-true] :
  < ENV : TimedComponent | tstate : < envdtimer : DelayTimer | value : 0, delay : 50 > >
 => 
  < ENV : TimedComponent | tstate : < envdtimer : DelayTimer | value : 50 >,
                           prov   : < ENV.persThereIn : Port | value : true > > .
\end{alltt}
The rule \verb|[env-true]| models the choice of \verb|ENV| to set the value of \verb|ENV.persThereIn| to \verb|true|; the rule \verb|[env-false]| is analogous.

For all instantaneous rewrite rules (except those introduced in
Sect.~\ref{sec:components}) we require that they are triggered by the
expiration of a timer which is indeed the case for all the rules in
our example. The advantage of this schema is that our specifications
are \emph{time-robust} \cite{DBLP:journals/entcs/OlveczkyM07a} for
which analysis techniques with the maximal time sampling strategy is
complete, i.e.~if there is a counterexample of a property to be
analyzed we will actually find it with the analysis technique.

\section{Analyzing in Real-Time Maude}\label{sec:analysis}

Real-Time Maude provides a variety of analysis techniques including
simulation through timed rewriting, untimed temporal logic model
checking, or (unbounded or time-bounded) search for reachability
analysis. However, for real-time specifications, timed properties
expressed in timed temporal logic are, of course, of great relevance,
e.g.~for a flight control system, changes in sensor information must
not only be reported eventually, but within a specific time bound. Up
to know, Real-Time Maude has lacked the ability to model check any
timed temporal logic formulas. In \cite{rtrts-Olveczky}, Lepri et al.\
show how to model check specific classes of timed temporal logic
formulas, expressed in metric temporal logic. In the same line as
\cite{rtrts-Olveczky}, we describe how to model check (different)
classes of metric temporal logic which will be shown useful for
analyzing our real-time specification for our case study.

Metric Temporal Logic (MTL) \cite{DBLP:journals/rts/Koymans90}
extends Linear Temporal Logic (LTL) \cite{kroegerTL} by
allowing to describe \emph{timed} properties of paths of a given
system which is useful for to specify time-critical systems. MTL is
more expressive than LTL, for instance, we can state the timed
property that some action should happen \emph{within some time
  bounds}, or that some property should always be satisfied
\emph{within an interval}. Formally, the syntax of MTL formulas is the
same as the syntax of LTL formulas, except for the until-operator
where a time interval is added. The formula $p \until{[b_1,b_2]} q$
states that $p \until{} q$ holds, i.e.~$p$ holds until $q$ holds, and
furthermore, $q$ occurs within the time interval $[b_1,b_2]$. Thus, in
MTL, time intervals are added to all derived operators like
$\always_{[b_1,b_2]}$ or $\eventually_{[b_1,b_2]}$.

MTL formulas $\phi$ are inductively defined as follows:
\[ \phi ::= true \mid{} p \mid{} \neg \phi \mid{} \phi_1 \land \phi_2
\mid{} \phi_1 \until{[b_1,b_2]} \phi_2\] where $p$ is a proposition
and for time intervals $[b_1,b_2]$ (and a given time domain $\mathbb
T$) we allow either $b_1,b_2 \in \mathbb T$, $b_1\leq b_2$ and $b_2 >
0$, or $t_1 \in \mathbb T$ and $t_2 = \infty$. Disjunction $\lor$ and
implication $\to$ are defined as usual.  $\eventually_{[b_1,b_2]}
\phi$ stands for $true \until{[b_1,b_2]} \phi$, and
$\always_{[b_1,b_2]}$ abbreviates $\neg (true \until{[b_1,b_2]} (\neg
\phi))$.  We will write $\until{\leq b}$ (and $\eventually_{\leq b}$,
$\always_{\leq b}$) if the lower bound is $0$.

We follow \cite{rtrts-Olveczky} in the notational conventions for
real-time rewrite theories: The set of states of a real-time rewrite
theory $\mathcal R = (\Sigma, E, IR, TR)$ is defined as the set of all
terms (modulo the equations in $E$) of type \verb|GlobalSystem|. A set
$\Pi$ of atomic propositions can be defined equationally (in a
protecting extension of $(\Sigma,E)$), and a labeling function $L_\Pi$
assigns to every state a finite set of propositions in $\Pi$ (cf.\
\cite{DBLP:journals/lisp/OlveczkyM07}).

Satisfaction of MTL formulas over timed paths of real-time rewrite
theories is defined as follows:

\begin{definition}[\cite{rtrts-Olveczky}]
  Let $\mathcal R$ be a real-time rewrite theory, $L_\Pi$ a labeling
  function on $\mathcal R$, and let $\pi = t_0 \xrightarrow{r_0} t_1
  \xrightarrow{r_1} \ldots$ be a timed path in $\mathcal R$. The
  satisfaction relation of an MTL formula $\phi$ for the path $\pi$ in
  $\mathcal R$ is then defined recursively as follows:
\begin{alignat*}{2}
  &\mathcal R,L_\Pi,\pi \vDash true&   \qquad &\mbox{always holds}\\
  &\mathcal R,L_\Pi,\pi \vDash p&       &\mbox{iff }p \in L_\Pi(t_0)\\
  &\mathcal R,L_\Pi,\pi \vDash \neg \phi& &\mbox{iff }\mathcal R,L_\Pi,\pi \not\vDash \phi\\
  &\mathcal R,L_\Pi,\pi \vDash \phi_1 \land \phi_2&       &\mbox{iff } \mathcal R,L_\Pi,\pi \vDash \phi_1 \mbox{ and } \mathcal R,L_\Pi,\pi \vDash \phi_2\\
  &\mathcal R,L_\Pi,\pi \vDash \phi_1 \until{[b_1,b_2]} \phi_2&
  &\mbox{iff there exists }j \in \mathbb N \mbox{ such that } \mathcal
  R,L_\Pi,\pi^j \vDash \phi_2\\ &&&\textstyle \mbox{ and }\mathcal
  R,L_\Pi,\pi \vDash \phi_1 \mbox{ for all }0 \leq i < j\mbox{, and }
  b_1 \leq \sum_{k=0}^{j-1}r_k \leq b_2
\end{alignat*}
For a state $t_0$ of sort \verb|GlobalSystem|, the satisfaction
relation of an MTL formula $\phi$ for the state $t_0$ in $\mathcal R$
is defined as follows:
\[\mathcal R, L_\Pi,t_0 \vDash \phi \Longleftrightarrow \forall \pi
\in \Paths(\mathcal R)_{t_0} \ .\ \mathcal R,L_\Pi,\pi \vDash \phi\]
\end{definition}

Real-Time Maude does currently not provide an MTL model
checker. However, in previous works, cf.,
e.g.~\cite{DBLP:conf/snpd/Olveczky08}, some simple MTL formulas could
already be model checked using the time-bounded search command or the LTL
model checker of Maude. In
\cite{rtrts-Olveczky}, Lepri et al.\ present an automatized analysis algorithm of
two important classes of MTL formulas, namely the bounded response
property $\always(p \to (\eventually_{\leq b} q))$ and the minimum
separation property $\always(p \to (p\weakuntil (\always_{\leq b} \neg
p)))$. We extend their ideas and present algorithms for two further
and more general classes of MTL formulas:
\begin{enumerate}
\item Generalized time-bounded response: $\always(\bigvee_{i\in I}(\eventually_{\leq b_i} q_i))$ for $I = \{1,2,\ldots,n\}\subset \mathbb N$ a finite set of indices, and
\item Time-Bounded safety: $\always(p \lor \always_{\leq b} q)$
\end{enumerate}
where $q_i$, $q$, and $p$ are all atomic propositions of $\Pi$.

In the following sections \ref{sec:eventually:algorithm} and
\ref{sec:always:algorithm}, we will describe the algorithm of the
transformation of the two classes of MTL formulas to corresponding LTL
formulas which are then model checked over the transformed rewrite
theory $\tilde{\mathcal R}$. So in each case, for each formula $\phi$
belonging to one of the classes, we will show that $\mathcal R, L_\Pi,
\pi \vDash \phi$ if and only if $\tilde{\mathcal R}, \tilde{L_\Pi},
\tilde{\pi} \vDash \tilde{\phi}$, hence it is shown how to modify the
rewrite theory $\mathcal R$ such that we can use the LTL model checker
of Maude to verify $\phi$.

\subsection{Model Checking MTL Formulas of the Form
  $\always(\bigvee_{i\in I}(\eventually_{\leq b_i}
  q_i))$}\label{sec:eventually:algorithm}

% Goal: We want to check the MTL-formula $\always(\bigvee_{i\in
%   I}(\eventually_{\leq r_i} q_i))$ by translating it to the
% LTL-formula $\always(\bigvee_{i\in I}(\eventually(q_i \land clock
% \leq r_i)))$.

For model checking MTL formulas of the form $\always(\bigvee_{i\in
  I}(\eventually_{\leq b_i} q_i))$ for a finite set $I$ of indices, we
add a single clock $c$ to the system state which will count the
elapsed time after each state in which no $q_i$ is satisfied. It is
indeed possible to restrict oneself to using one single clock by
leveraging the observation that given a sequence of states $t_1,
\ldots, t_n$ satisfying no $q_i$ within the time interval $\max\{b_i
\mid i \in I\}$, the first state $s_1$ determines the deadlines
$b_i$, and hence we can set the single clock $c$ to zero and start it
in state $s_1$: the first occurrence of a $q_i$-satisfying state $s_w$
will also witness the validity of $\always(\bigvee_{i\in
  I}(\eventually_{\leq b_i} q_i))$ in all $s_i$ states between $s_1$
and $s_w$ (i.e.~state $s_i$ satisfies the formula for all $1 \leq i
\leq w$). The clock is switched off in state $s_w$, and can be
switched back on if another state satisfying none of the $q_i$'s is
found after the witness state $s_w$, and restart the counting
process. In summary, to model check the MTL formula
$\phi\equiv\always(\bigvee_{i\in I}(\eventually_{\leq b_i} q_i))$ for
a path $\pi$ in a real-time rewrite theory $\mathcal R$, the following
steps are necessary: First, to $\mathcal R$ a class \verb|Clock|,
modeling the clock, and corresponding equations are added; second,
$\phi$ is translated to $\tilde{\phi}\equiv \always(\bigvee_{i\in
  I}(\eventually(q_i \land clock \leq b_i)))$ where $clock$ is an
atomic proposition which refers to the current time value of the
clock; the rewrite rules are transformed to adequately take into
account the propositions and the clock behavior. The transformation algorithm,
called \emph{$\eventually$-transformation} in the following, comprises four steps:
\begin{enumerate}
\item A class modeling the clock is added:
\small
\begin{verbatim}
  sort ClockStatus .
  ops on off : -> ClockStatus [ctor] .
  class Clock | clock : Time, status : ClockStatus .
\end{verbatim}
\normalsize
\item The initial state $\{t_0\}$ is modified by adding a clock object
  such that the new initial state is
\begin{alltt}\small
  \{\(t\sb{0}\) <\(\ c\) : \!\!Clock | clock :\!\! 0, status :\(\ x\) >\}
\end{alltt}
where $c$ is a constant of sort \verb|Oid| and $x$ is \verb|off| if $\{t_0\} \vDash \bigvee_{i\in I} q_i$,
  else \verb|on|.
\item The functions \verb|delta| (modeling the effect of time elapse on a configuration) and \verb|mte| (computing the maximum time elapse for a configuration) are extended for the newly introduced clock object as follows:
\begin{alltt}\small
  eq delta(< \(c\) : Clock | status : on, clock : T >, T') =
           < \(c\) : Clock | clock : if T <= \(b\sb{max}\) then min(T+T',\(b\sb{max}\)+1) else T fi > .
  eq delta(< \(c\) : Clock | status : off >, T') = < \(c\) : Clock | > .
  eq mte(< \(c\) : Clock | >) = INF .
\end{alltt}
where $b_{max} = \max\{b_i \mid i \in I\}$.
\item Instantaneous rewrite rules are modified such that the clock is
  switched on and off depending on the target state. Each
  instantaneous rule \texttt{$t$ => $t'$ if $cond$} or \texttt{\{$t$\}
    => \{$t'$\} if $cond$} in $\mathcal R$ is replaced by the
  following four rules (where \verb|REST| is a new variable of type
  \verb|Configuration|):
\begin{enumerate}[Rule (1):]
\item If the clock is off then the clock stays off if at least one of the $q_i$'s is satisfied.
\begin{alltt}\small
  \{\(t\) REST < \(c\) : Clock | status : off >\}
    => \{\(t'\) REST < \(c\) : Clock | >\}
  if (modelCheck(\{\(t'\) REST\},\(q\sb{1}\)) == true  or
       \(...\)
      modelCheck(\{\(t'\) REST\},\(q\sb{n}\)) == true)  and \(cond\) .
\end{alltt}
\item If the clock is off then the clock is switched on if none of the $q_i$'s is satisfied.
\begin{alltt}\small
  \{\(t\) REST < \(c\) : Clock | status : off >\}
    => \{\(t'\) REST < \(c\) : Clock | clock : 0, status : on >\}
  if (not (modelCheck(\{\(t'\) REST\},\(q\sb{1}\)) == true  or
            \(...\)
           modelCheck(\{\(t'\) REST\},\(q\sb{n}\)) == true))  and \(cond\) .
\end{alltt}
\item If the clock is on then the clock stays on if for all $i \in I$, either $q_i$ is not satisfied or the time bound is already exceeded.
\begin{alltt}\small
  \{\(t\) REST < \(c\) : Clock | clock : T, status : on >\}
    => \{\(t'\) REST < \(c\) : Clock | >\}
  if (not ((modelCheck(\{\(t'\) REST\},\(q\sb{1}\)) == true and (T <= \(b\sb{1}\)))  or
           \(\ldots\)  or
           (modelCheck(\{\(t'\) REST\},\(q\sb{n}\)) == true and (T <= \(b\sb{n}\)))))  and \(cond\) .
\end{alltt}
\item If the clock is on then the clock is switched off if at least one of the $q_i$'s is satisfied within its corresponding time bound $b_i$.
\begin{alltt}\small
  \{\(t\) REST < \(c\) : Clock | clock : T, status : on >\}
    => \{\(t'\) REST < \(c\) : Clock | clock : 0, status : off >\}
  if (modelCheck(\{\(t'\) REST\},\(q\sb{1}\)) == true and (T <= \(b\sb{1}\)))  or
     \(\ldots\)  or
     (modelCheck(\{\(t'\) REST\},\(q\sb{n}\)) == true and (T <= \(b\sb{n}\)))  and \(cond\) .
\end{alltt}
\end{enumerate}
\end{enumerate}
Thus, by the above steps 1. to 4. we obtain a real-time rewrite theory
$\tilde{\mathcal R}$, a labeling function $\tilde{L}_\Pi$ which is
adapted to the transformed state space while the labeling remains
unchanged (i.e.~$L_\Pi(\{t\}) = \tilde{L}_\Pi(\{\tilde{t}\ o_{clock}\})$ where
$o_{clock}$ is the added clock), and $\{\tilde{t}_0\}$ is the transformed
initial state.

Finally, for model checking the MTL formula we need to add an atomic
proposition stating that the current clock value is less or equal than
a given time value $r$.
\begin{alltt}\small
op clockLeq : Time -> Prop [ctor] .
eq \{< \(c\) : Clock | clock : t, status : s > REST\} |= clockLeq(b) = (t <= b) .
\end{alltt}
The MTL formula $\always(\bigvee_{i\in I}(\eventually_{\leq b_i}
  q_i))$ can then be model checked using Real-Time Maude's untimed LTL model checking features, i.e.~we check whether the transformed formula holds by invoking\\[0.1cm]
\small
\verb|(mc |$\{\tilde{t}_{0}\}$ \verb+ |=u+\\
\verb|    [] ( (<>(| $q_{1}$ \verb|/\ clockLeq(| $b_{1}$ \verb|)))|\\
\verb|          \/ | $\ldots$ \verb| \/|\\
\verb|         (<>(| $q_{n}$ \verb|/\ clockLeq(| $b_{n}$ \verb|))) )  .)|\\
\normalsize
which precisely is $\tilde{\mathcal R},\tilde{L}_\Pi,\{\tilde{t}_0\} \vDash \always(\bigvee_{i\in
  I}(\eventually(q_i \land clock \leq b_i)))$.
\paragraph{Proof of Correctness of the Transformation.}
\begin{lemma}[cf. \cite{rtrts-Olveczky}]\label{lem:bisimilar}
Let $\mathcal R$ be a real-time rewrite theory, $L_\Pi$ with $q_i \in \Pi$ for all $i \in I$ a labeling function for $\mathcal R$, and let $\{t_0\}$ be an initial state for $\mathcal R$. Let $\tilde{\mathcal R}$, $\tilde{L_\Pi}$, and $\{\tilde{t_0}\}$ be the result of the $\eventually$-transformation applied to $\mathcal R$, $L_\Pi$, and $t_0$.
Then for each path $\{t_0\} \xrightarrow{r_0} \{t_1\} \xrightarrow{r_1} \ldots$ in $\mathcal R$ there is a path $\{\tilde{t_0}\} \xrightarrow{r_0} \{\tilde{t_1}\} \xrightarrow{r_1} \ldots$ in $\tilde{\mathcal R}$ such that, for all $i \geq 0$, there exists $t'_i$ with $\tilde{t_i} = t_i t'_i$, and vice versa.
\end{lemma}

\begin{proof}
We have to show that the transformation does not modify the original timed behavior. This is ensured by the following facts:
\begin{itemize}
\item Adding the clock class and a clock object to the initial state does not affact the original part of the state, and moreover, the timed behavior of the original system is not affected by the newly introduced clock since \verb|mte| of the clock evaluates to \verb|INF|.
\item The transformation replaces each rewrite rules by a number of rules with additional conditions. However, for each (extended) state to which the original rule is applicable, there is exactly one new rule applicable, and furthermore, the new rules treat the original state part as the original rule.
\end{itemize}
It follows that the original timed behavior is not modified, in particular, no original paths are blocked by the new rules, and conversely, new rules yield the same result for the original part of the state.
% Given a path $\{t_0\} \xrightarrow{r_0} \{t_1\} \xrightarrow{r_1} \ldots$ in $\mathcal R$, we define
% \[\tilde{t_i} = t_i\ \texttt{<$\ c$ :\!\! Clock | clock :$\ x_i$, status :$\ y_i$ >}.\]
% It can be shown that $x_0$ and $y_0$ in the initial state $\tilde{t_0}$ are uniquely determined, and moreover, for all $i > 0$, the rules in $\tilde{\mathcal R}$ guarantee that $x_i$ and $y_i$ in $\tilde{t_i}$ are also uniquely determined.
\end{proof}

\begin{theorem}
Let $\mathcal R$ be a real-time rewrite theory, $L_\Pi$ a labeling function for $\mathcal R$ with $q_i \in \Pi$ for all $i \in I$, and $\{t_0\}$ an initial state of $\mathcal R$. Let $\tilde{\mathcal R}$, $\tilde{L_\Pi}$, and $\{\tilde{t_0}\}$ be the result of the $\eventually$-transformation applied to $\mathcal R$, $L_\Pi$, and $\{t_0\}$. Then the following equivalence holds:
\[ \mathcal R, L_\Pi, \{t_0\} \vDash \always\bigvee_{i\in I} \eventually_{\leq b_i} q_i
\quad \Longleftrightarrow \quad 
\tilde{\mathcal R}, \tilde{L_\Pi}, \{\tilde{t_0}\} \vDash \always \bigvee_{i\in I} \eventually(q_i \land clock \leq b_i)))
\]

%\begin{enumerate}
%\item
%$\mathcal R, L_\Pi, \{t_0\} \vDash \always\bigvee_{i\in I} \eventually_{\leq r_i} q_i \implies \tilde{\mathcal R}, \tilde{L_\Pi}, \{\tilde{t_0}\} \vDash \always(clock(c) \leq \max_{i\in I}r_i)$, and
%%\tilde{\mathcal R}, \tilde{L_\Pi}, \{\tilde{t_0}\} \vDash \always \bigvee_{i\in I} \eventually ( q_i \land clock(c) \leq r_i )
%\item $\tilde{\mathcal R}, \tilde{L_\Pi}, \{\tilde{t_0}\} \vDash (\always\bigvee_{i\in I} \eventually q_i) \land (\always(clock(c) \leq \max_{i\in I}r_i)) \implies \mathcal R, L_\Pi, \{t_0\} \vDash \always\bigvee_{i\in I} \eventually_{\leq r_i} q_i$,\\
%where $clock(c)$ denotes the value of the \verb|clock| attribute of the clock object $c$.
%\end{enumerate}
\end{theorem}
\begin{proof}

  ``$\Longrightarrow$'': Assume $\tilde{\mathcal R}, \tilde{L_\Pi},
  \{\tilde{t_0}\} \not\vDash \always \bigvee_{i\in I} \eventually(q_i
  \land clock \leq b_i)))$, we show $\mathcal R, L_\Pi, \{t_0\}
  \not\vDash \always\bigvee_{i\in I} \eventually_{\leq b_i} q_i$.  Let
  $\tilde{\pi} = \{\tilde{t}_0\} \xrightarrow{r_0} \{\tilde{t}_1\}
  \xrightarrow{r_1} \ldots$ be a path in $\tilde{\mathcal R}$ which
  does not satisfy $\always \bigvee_{i\in I} \eventually(q_i \land
  clock \leq b_i)))$.
  % By Lemma there exists a path $\pi$ in $\mathcal R$.  By
  % Lemma~\ref{lem:b}, it suffices to show that there exists a pat we
  % have to show $\exists j \geq 0.\forall i \in I.\exists k_i >
  % j.\forall j \leq l \leq k_i. \pi^l \not\vDash q_i \land
  % \sum_{l=j}^{k_i-1} b_l > r_i$.
  By definition of $\vDash$ we know that there exists $j \geq 0$ such
  that $\tilde{\pi}^j \not\vDash \bigvee_{i\in I} \eventually(q_i
  \land clock \leq b_i))$, i.e.
  \begin{equation*}\tag{1}
    \forall i \in I.\forall k \geq j. (\tilde{\pi}^k \not\vDash q_i) \lor (\tilde{\pi}^k \not \vDash clock \leq b_i).
  \end{equation*}
  Let $j \geq 0$ be the smallest index satisfying (1), and therefore, if $j=0$ then the clock status is \texttt{off}, otherwise $j > 0$ and 
  $\tilde{\pi}^{j-1} \vDash \bigvee_{i\in I} \eventually(q_i \land
  clock \leq b_i))$. It follows that there exists $i \in I$ such that
  $\tilde{\pi}^{j-1} \vDash q_i \land clock \leq b_i$. Rule (4) ensures that as soon as this formula is satisfied, the clock status is \texttt{off}, hence the clock status in $\{\tilde{t}_{j-1}\}$ is \texttt{off}, too. It follows that the
  rewrite step from $\{\tilde{t}_{j-1}\}$ to $\{\tilde{t}_j\}$ is an
  instantaneous step of the form of rule (2) which sets the clock
  status to \verb|on| and the clock value to $0$. Furthermore, in both cases the
  clock can only be switched off by rule (4) which can never be
  applied because of the condition $\bigvee_{i\in I} q_i \land
  clock \leq b_i$ which is -- by assumption -- not satisfied. We
  can conclude that, from state $\{\tilde{t}_j\}$ on, the clock is
  continuously on and the clock value equals the elapsed time since
  $\{\tilde{t}_j\}$, i.e.~the clock value is the sum of the durations
  of the applied tick rules since $\{\tilde{t}_j\}$.\footnote{The
    clock value will not be greater than $\max_{i\in I} b_i + 1$.}
  Therefore, for all $i \in I$ and for all $k \geq j$, $\tilde{\pi}^k
  \vDash clock > b_i$ if and only if $\sum_{l=j}^{k-1} r_l >
  b_i$. From (1) it follows
  \begin{equation*}\tag{2}
    \forall i \in I.\forall k \geq j. (\tilde{\pi}^k \not\vDash q_i) \lor \left(\sum_{l=j}^{k-1} r_l > b_i\right).
  \end{equation*}
  Hence from (2) we can conclude that $\tilde{\pi}^k \not\vDash q_i$
  for all $k\geq j$ such that $\sum_{l=j}^{k-1} r_l \leq b_i$. This
  implies $\tilde{\pi}^j \not\vDash \bigvee_{i\in I}\eventually_{\leq
    b_i}q_i$, and then $\tilde{\pi} \not\vDash \always \bigvee_{i\in
    I}\eventually_{\leq b_i}q_i$. By Lemma~\ref{lem:bisimilar}, there
  exists a unique path $\pi$ with initial state $\{t_0\}$ for which
  $\pi \not\vDash \always\bigvee_{i\in I}\eventually_{\leq b_i}
  q_i$. Finally, it follows $\mathcal R, L_\Pi, \{t_0\} \not\vDash
  \always\bigvee_{i\in I} \eventually_{\leq b_i} q_i$ which was to be
  shown.

  ``$\Longleftarrow$'': Assume $\mathcal R, L_\Pi, \{t_0\} \not\vDash
  \always\bigvee_{i\in I} \eventually_{\leq b_i} q_i$, we show
  $\tilde{\mathcal R}, \tilde{L_\Pi}, \{\tilde{t_0}\} \not\vDash
  \always \bigvee_{i\in I} \eventually(q_i \land clock \leq
  b_i)$. Let $\pi = \{t_0\} \xrightarrow{r_0} \{t_1\}
  \xrightarrow{r_1} \ldots$ be a path in $\mathcal R$, by
  Lemma~\ref{lem:bisimilar}, we also have a path $\tilde{\pi} =
  \{\tilde{t}_0\} \xrightarrow{r_0} \{\tilde{t}_1\} \xrightarrow{r_1}
  \ldots$ in $\tilde{\mathcal R}$. By assumption, $\pi$ and hence also
  $\tilde{\pi}$ do not satisfy $\always\bigvee_{i\in I}
  \eventually_{\leq b_i} q_i$, i.e. there exists $j \geq 0$ such that
\begin{equation*}
  \forall i \in I.\forall k \geq j.(\tilde{\pi}^k \not\vDash q_i) \lor \left(
    \sum_{l=j}^{k-1}r_l > b_i\right).\tag{3}
\end{equation*}
Let $j \geq 0$ be the minimal index satisfying (3). We show that in $\{\tilde{t}_j\}$ the clock status is \verb|on|. If $j = 0$ then the clock status is \verb|on| by definition of the initial state. Now assume $j > 0$. Then in state
$\{\tilde{t}_{j-1}\}$ it must hold $\tilde{\pi}^{j-1} \vDash q_i$
% and $\sum_{l=j}^{k-1}r_l \leq b_i$
for some $i \in I$.  So the clock status in $\{\tilde{t}_{j-1}\}$ is
\verb|off| and the clock value is $0$ because otherwise, if the clock status was \verb|on|, there would exist a state before $j$ not satisfying (3) and hence contradicting our assumption. From (3) it follows that the rewrite
step from $j-1$ to $j$ is an instantaneous rewrite step, switching the
clock on (with clock value $0$). Since the clock cannot be switched
off (the conditions of rule (4) are never met from $\tilde{\pi}^j$
on), the durations of the tick steps since $\tilde{t}_j$ and the clock
value are equal. It follows that
\begin{equation*}
  \forall i \in I.\forall k \geq j.(\tilde{\pi}^k \not\vDash q_i) \lor \left(\tilde{\pi}^k \not\vDash
    clock \leq b_i\right)
\end{equation*}
which implies $\tilde{\mathcal R}, \tilde{L}_\Pi, \{\tilde{t}_0\}
\not\vDash \always \bigvee_{i\in I} \eventually(q_i \land clock
\leq b_i)$ which was to be shown.
% $\exists j \geq 0.\forall i \in I.\left(\exists k > j.(\forall j
%   \leq l \leq k. \pi^l \not\vDash q_i) \land \sum_{l=j}^{k-1} r_l >
%   r_i\right)\lor\left(\forall l \geq j.\pi^l \not\vDash q_i
% \right)$. By Lemma~\ref{lem:a}, we have a unique path $\tilde{\pi}$
% in $\mathcal R$ for which the same formula holds. Let $j \geq 0$ be
% the minimal such index. Similar to the proof for the other direction
% above, minimality of $j$ implies that the last rewrite step
% $\{\tilde{t}_{j-1}\}\xrightarrow{r_{j-1}}\{\tilde{t}_j\}$ has been
% an instantaneous step which has switched on the clock, and the clock
% value in $\tilde{t}_j$ is $0$. The rules in $\tilde{\mathcal R}$
% ensure that the clock value in each state $\tilde{t}_k$ (for $k \geq
% j$) is always $\sum_{l=j}^{k-1}r_l$, i.e.~the clock value always
% equals the elapsed time since $j$, and the clock is never switched
% off. Therefore, $\tilde{\pi} \not\vDash \bigvee_{i\in I}
% \eventually(q_i \land clock(c) \leq r_i)$ which was to be shown.
% because beyond time $\max_{i\in I}r_i$, the formula $clock(c) \leq
% r_i$ is never satisfied.
\end{proof}

\subsection{Model Checking MTL Formulas of the Form $\always(p \lor
  \always_{\leq b} q)$}\label{sec:always:algorithm}

% Goal: We want to check the MTL-formula $\always(p \lor \always_{\leq
%   r} q)$ by translating it to the LTL-formula $\always(p \lor (q\
% \mathsf{W}\ (clock(c) > r)))$.

For model checking MTL formulas of the form $\always(p \lor
\always_{\leq b} q)$, we add a single clock which counts the minimum
time that $q$ needs to be true once $p$ became false. Here, we use the
observation that if $p$ was false at $t_1$ and becomes false again
between $t_1$ and $t_1 + b$, say at $t_2$, $q$ must additionally hold
until $t_2 + b$. Hence, it is valid to reset the clock at $t_2$ and
thereby enforce that $q$ must hold true for $b$ more time units.
So to model check the MTL formula
$\phi\equiv\always(p \lor \always_{\leq b} q)$ for
a path $\pi$ in a real-time rewrite theory $\mathcal R$ with labeling
function $L_\Pi$, the following steps are necessary: First, to
$\mathcal R$ a class \verb|Clock|, modeling the clock, and
corresponding equations are added; second, $\phi$ is translated to
$\tilde{\phi}\equiv \always(p \lor (q \weakuntil (clock > b)))$
where $clock$ is an atomic proposition which refers to the current
time value of the clock; the rewrite rules are transformed to
adequately take into account the propositions and the clock
behavior. The transformation, which we will call \emph{$\always$-transformation}
in the following, proceeds as follows.
\begin{enumerate}
\item A class modeling the clock is added
 (analogous to the $\eventually$-transformation):
\small
\begin{verbatim}
sort ClockStatus .
ops on off : -> ClockStatus [ctor] .
class Clock | clock : Time, status : ClockStatus .
\end{verbatim}
\normalsize
\item The initial state $\{t_0\}$ is modified by adding a clock object
  such that the new initial state is \[\small\texttt{\{$t_0$\ <$\ c$\ :\ Clock\
    |\ clock\ :\ 0,\ status\ :\ off\ >\}}\]
\item The functions \verb|delta| and \verb|mte| are extended for the
  newly introduced class \verb|Clock| as follows (again analogous to the
$\eventually$-transformation):
\begin{alltt}\small
eq delta(< \(c\) : Clock | status : on, clock : T >, T') =
         < \(c\) : Clock | clock : if T <= \(b\) then min(T + T',\(b\)+1) else T fi > .
eq delta(< \(c\) : Clock | status : off >, T') = < \(c\) : Clock | > .
eq mte(< \(c\) : Clock | >) = INF .
\end{alltt}
\item  Instantaneous rewrite rules are modified such that the clock is
  switched on and off depending on the target state. Each
  instantaneous rule \texttt{$t$ => $t'$ if $cond$} or \texttt{\{$t$\}
    => \{$t'$\} if $cond$} in $\mathcal R$ is replaced by the
  following four rules (where \verb|REST| is a new variable of type
  \verb|Configuration|):
\begin{enumerate}[Rule (1):]
\item If in the next state the formula $\neg p \lor \neg q$ is satisfied, or in
the previous state the formula $p \lor \neg q$ is satisfied, then the clock stays or is switched off.
\begin{alltt}\small
  \{\(t\) REST < \(c\) : Clock | >\}
    => \{\(t'\) REST < \(c\) : Clock | clock : 0, status : off >\}
  if (modelCheck(\{\(t'\) REST\}, ~ \(p \verb+\/+\) ~ \(q\)) == true  or
      modelCheck(\{\(t\) REST\},\(p \verb+\/+\) ~ \(q\)) == true)  and \(cond\) .
\end{alltt}
\item If the clock is off, it only gets switched on if in the previous state
$\neg p \land q$ was satisfied and in the next state $p \land q$ is satisfied.
So the clock begins to count if there was a state where $p$ was not true (so we need to
look for an interval of length $\geq r$ where $q$ always holds) and in the next state
$p$ is true (so the formula $\always(p \lor \always_{\leq b} q)$ is satisfied).
\begin{alltt}\small
  \{\(t\) REST < \(c\) : Clock | status : off >\}
    => \{\(t'\) REST < \(c\) : Clock | clock : 0, status : on >\}
  if (not (modelCheck(\{\(t'\) REST\}, ~ \(p \verb+\/+\) ~ \(q\)) == true  or
           modelCheck(\{\(t\) REST\},\(p \verb+\/+\) ~ \(q\)) == true))  and \(cond\) .
\end{alltt}
\item If the clock is on and in the next state the formula $p \land q$ is satisfied
then the clock stays on. The clock is only on if we are looking for an interval of length
$\geq b$ such that $q$ is satisfied, so we can safely go on with counting the advanced time
since we do not ``miss'' any counterexample since $p$ is satisfied in the next state.
\begin{alltt}\small
  \{\(t\) REST < \(c\) : Clock | status : on >\}
    => \{\(t'\) REST < \(c\) : Clock | >\}
  if modelCheck(\{\(t'\) REST\},\(p \verb+/\+ q\)) == true  and \(cond\) .
\end{alltt}
\end{enumerate}
\end{enumerate}

Thus, by the above steps 1.\ to 4.\ we obtain a real-time rewrite theory
$\tilde{\mathcal R}$, a labeling function $\tilde{L}_\Pi$ which is
adapted to the transformed state space while the labeling remains
unchanged (i.e.~$L_\Pi(\{t\}) = \tilde{L}_\Pi(\{\tilde{t}\ o_{clock}\})$ where
$o_{clock}$ is the added clock), and $\{\tilde{t}_0\}$ is the transformed
initial state.

Finally, for model checking the MTL formula we need to add an atomic
proposition stating that the current clock value is less or equal than
a given time value $b$.
\begin{alltt}\small
op clockLeq : Time -> Prop [ctor] .
eq \{< \(c\) : Clock | clock : t, status : s > REST\} |= clockLeq(\(b\)) = (t <= \(b\)) .
\end{alltt}
The MTL formula $\always(p \lor \always_{\leq b} q)$ can then be model checked using Real-Time Maude's untimed LTL model checking features, i.e.~we check whether the transformed formula holds by invoking
\begin{alltt}\small
(mc \{\(\tilde{t}\sb{0}\)\} |=u [] (\(p \verb+\/+ q\) W (not clockLeq(\(b\))))) .)
\end{alltt}
%which precisely is $\tilde{\mathcal R},\tilde{L}_\Pi,\{\tilde{t}_0\} \vDash \always(p \lor \always_{\leq r} q)$.

\paragraph{Proof of Correctness of the Transformation.}

\begin{lemma}[cf.\ \cite{rtrts-Olveczky}]\label{lem:always:paths}
Let $\mathcal R$ be a real-time rewrite theory, $L_\Pi$ with $p,q \in \Pi$ a labeling function for $\mathcal R$, and let $\{t_0\}$ be an initial state for $\mathcal R$. Let $\tilde{\mathcal R}$, $\tilde{L}_\Pi$, and $\{\tilde{t}_0\}$ be the result of the $\always$-transformation applied to $\mathcal R$, $L_\Pi$, and $t_0$.
Then for each path $\{t_0\} \xrightarrow{r_0} \{t_1\} \xrightarrow{r_1} \ldots$ in $\mathcal R$ there is a path $\{\tilde{t}_0\} \xrightarrow{r_0} \{\tilde{t}_1\} \xrightarrow{r_1} \ldots$ in $\tilde{\mathcal R}$ such that, for all $i \geq 0$, there exists $t'_i$ with $\tilde{t_i} = t_i t'_i$, and vice versa.
\end{lemma}

\begin{proof}
Very similar to the proof of Lem.~\ref{lem:bisimilar}.
% Given a path $\{t_0\} \xrightarrow{r_0} \{t_1\} \xrightarrow{r_1} \ldots$ in $\mathcal R$, we define
% \[\tilde{t_i} = t_i\ \texttt{<$\ c$ :\!\! Clock | clock :$\ x_i$, status :$\ y_i$ >}.\]
% It can be shown that $x_0$ and $y_0$ in the initial state $\tilde{t}_0$ are uniquely determined, and moreover, for all $i > 0$, the rules in $\tilde{\mathcal R}$ guarantee that $x_i$ and $y_i$ in $\tilde{t}_i$ are also uniquely determined.
\end{proof}

% \begin{lemma}\label{lem:always:equiv}
% Let $\mathcal R$ be a real-time rewrite theory, $L_\Pi$ a labeling function for $\mathcal R$ with $q_i \in \Pi$ for all $i \in I$, and $\pi = \{t_0\} \xrightarrow{r_0} \{t_1\} \xrightarrow{r_1} \ldots$ a path of $\mathcal R$. Then the following equivalence is satisfied:
% \[\bigg[\pi \not\vDash \always\left(p \lor \always_{\leq r} q\right)\bigg]
% \quad
% \Longleftrightarrow
% \quad
%  \left[ \exists j \geq 0.\left(\pi^j \not\vDash p\right) \land \left(\exists k \geq j.(\pi^k \not\vDash q) \land \sum_{l=j}^{k-1} r_l \leq r\right)\right].\]
% \end{lemma}

% \begin{proof}
% Immediately by the definition of $\vDash$.
% \end{proof}

\begin{theorem}
  Let $\mathcal R$ be a real-time rewrite theory, $L_\Pi$ a labeling
  function for $\mathcal R$ with $p,q \in \Pi$, and
  $\{t_0\}$ an initial state of $\mathcal R$. Let $\tilde{\mathcal
    R}$, $\tilde{L}_\Pi$, and $\{\tilde{t}_0\}$ be the result of the
  $\always$-transformation applied to $\mathcal R$, $L_\Pi$, and
  $\{t_0\}$. Then the following equivalence holds:
  \[ \mathcal R, L_\Pi, \{t_0\} \vDash \always(p \lor \always_{\leq b}
  q) \quad \Longleftrightarrow \quad \tilde{\mathcal R},
  \tilde{L}_\Pi, \{\tilde{t}_0\} \vDash \always(p \lor (q \weakuntil
  (clock > b))).
  \]
\end{theorem}

\begin{proof}
  ``$\Longrightarrow$'': Let $\tilde{\pi} = \{\tilde{t}_0\} \xrightarrow{r_0}
  \{\tilde{t}_1\} \xrightarrow{r_1} \ldots$ be a path in
  $\tilde{\mathcal R}$. Assume $\tilde{\mathcal R}, \tilde{L_\Pi},
  \{\tilde{t_0}\} \not\vDash \always(p \lor (q \weakuntil (clock > b)))$, and we show $\mathcal R, L_\Pi, \{t_0\} \not\vDash
  \always(p \lor \always_{\leq b} q)$. By assumption,
  \begin{align*}\exists j \geq 0.(\tilde{\pi}^j \not\vDash
    p)\land(\exists k\geq j.&(\tilde{\pi}^k\not\vDash q) \land
    (\tilde{\pi}^k \not\vDash clock > b)\land\\
    &(\forall l. j\leq l < k \Rightarrow (\tilde{\pi}^l \vDash q)
    \land (\tilde{\pi}^l \not\vDash clock > b))).
  \end{align*}
  Let $j \in \mathbb N$ be the minimal index satisfying the above formula.
  We know $\tilde{\pi}^j\not\vDash p$, if in addition
  $\tilde{\pi}^j\not\vDash q$ then we are finished because obviously $\tilde{\pi}^j\not\vDash p \lor \always_{\leq b}
  q$. So assume
  $\tilde{\pi}^j\vDash q$. If $j=0$ then the clock status is
  \verb|off| and the clock value is $0$; if $j>0$ then the last
  rewrite step
  $\{\tilde{t}_{j-1}\}\xrightarrow{r_{j-1}}\{\tilde{t}_j\}$ must have been
  the rule (1) since $\neg p$ is satisfied in $\tilde{t}_j$, i.e.~the clock status in $\tilde{t}_j$ is
  \verb|off| and the clock value is $0$. Hence, in any case, we know
  that in $\tilde{t}_j$ the clock status is \verb|off| and the
  clock value is $0$. Let $k > j$ % TODO really k > j or should it be k \geq j?  
  be the minimal index such that
  $\tilde{\pi}^k \not \vDash q$, $\tilde{\pi}^k \not\vDash clock >
  b$, and $\forall l. j\leq l < k \Rightarrow (\tilde{\pi}^l \vDash q)
  \land (\tilde{\pi}^l \not\vDash clock > b)$; such a $k$ exists by assumption.

  We know that $p$ is not satisfied in $\tilde{t}_j$. If $p$ is not
  satisfied for all states $\tilde{t}_m$ for $j\leq m < k$ then we can
  just take the state $\tilde{t}_{k-1}$ as an counterexample, that is
  $\tilde{\pi}^{k-1} \not\vDash p \lor \always_{\leq b} q$.  So assume
  that there exists a maximal $m$ with $j \leq m < k$ such that
  $\tilde{t}_m$ does not satisfy $p$. Furthermore we can assume that
  between $m$ and $k$ there is at least one tick rule since otherwise,
  in state $\tilde{t}_m$ the proposition $p$ is not satisfied, but
  also $q$ is not satisfied in $\tilde{t}_k$ which is reachable in
  zero time. It follows that in this case $\tilde{\pi}^m$ does not
  satisfy $p \lor \always_{\leq b} q$ and hence there must be a tick
  rule between $\tilde{t}_m$ and $\tilde{t}_k$, and moreover, $m \leq
  k - 3$ (after the $m$th state there must be an instantaneous step
  changing $\neg p$ to $p$, one application of the tick rule, and an
  instantaneous step changing $q$ to $\neg q$).

  It follows that the rewrite step
  $\{\tilde{t}_m\}\xrightarrow{r_m}\{\tilde{t}_{m+1}\}$ switches the
  clock on ($\tilde{t}_m$ satisfies $\neg p$ and $q$, and
  $\tilde{t}_{m+1}$ satisfies $p$ and, by the above observation that
  $m \leq k-3$, also $q$). From the state $\tilde{t}_{m+1}$ on the
  clock counts and for all subsequent states up to $\tilde{t}_k$ the
  clock value equals the duration of the tick rules between
  $\tilde{t}_m$ and $\tilde{t}_k$.  Since, by assumption,
  $\tilde{\pi}^k \vDash clock \leq b$, it follows $\sum_{l=m}^{k-1}r_l
  \leq b$ and we can conclude that in $\tilde{\pi}^m$ the formula $p
  \lor \always_{\leq r} q$ is not satisfied: $p$ is not satisfied in
  $\tilde{t}_m$ and moreover, $q$ does not hold for all states
  reachable within time $b$. Thus $\tilde{\pi} \not\vDash \always(p
  \lor \always_{\leq b} q)$, and since $\tilde{\pi}$ was an arbitrary
  path in $\tilde{\mathcal R}$ with initial state $\{\tilde{t}_0\}$ it
  follows from Lemma~\ref{lem:always:paths} that $\mathcal R, L_\Pi,
  \{t_0\} \not\vDash \always(p \lor \always_{\leq b} q)$.

%  We can assume that for all $s$ with $j < s \leq k$,
%  $\sum_{l=j}^{s-1}r_l < \sum_{l=j}^{k-1}r_l$ implies $\tilde{\pi}^s
%  \vDash p$. If there is such an $s$ with $\tilde{\pi}^s \not\vDash
%  p$, then the clock value is set to $0$ while staying in the clock
%  status \verb|ready| which means that we can take the index $s$ as
%  our new $j$. Thus the rewrite step $\{\tilde{t}_j\}
%  \xrightarrow{r_j} \{\tilde{t}_{j+1}\}$ is an instantaneous step with
%  $\{\tilde{t}_j\} \not\vDash p$ and $\{\tilde{t}_{j+1}\} \vDash p$
%  (i.e.~by rule (7), clock status is switched to \verb|on|), and $p$
%  stays satisfied until the last tick step before $k$. Hence the clock
%  value in $\{\tilde{t}_s\}$ equals $\sum_{l=j}^{s-1}r_l$ and the
%  clock status is always \verb|on| since only tick steps and
%  instantaneous rules of the form of rule (4) can occur. After the
%  last tick step and before the last (instantaneous) step, the clock
%  value can only be set to $0$ (and status to \verb|ready|) by one of
%  the rules (1), (5), or (6), but if there was no state
%  $\{\tilde{t_m}\}$ after the last tick step satisfying $q$ and
%  $clock(c)>r$ then setting $clock(c)$ to $0$ does not change anything
%  in this respect.
%
%  Thus we have shown that
%  \[\exists j \geq 0.(\tilde{\pi}^j \not\vDash p) \land \left(\exists
%    k \geq j.(\tilde{\pi}^k\not\vDash q) \land \sum_{l=j}^{k-1}r_l
%    \leq r\right)\] and we can % by Lemma~\ref{lem:always:equiv} we can

  ``$\Longleftarrow$'': Assume $\mathcal R, L_\Pi, \{t_0\} \not\vDash
  \always(p \lor \always_{\leq b} q)$, and we show $\tilde{\mathcal
    R}, \tilde{L}_\Pi, \{\tilde{t}_0\} \not\vDash \always(p \lor (q
  \weakuntil (clock > b)))$. Let $\pi = \{t_0\}
  \xrightarrow{r_0} \{t_1\} \xrightarrow{r_1} \ldots$ be a path in
  $\mathcal R$. By assumption %Lemma~\ref{lem:always:equiv} and
	we know
  \begin{equation*}
    \exists j \geq 0.\left(\tilde{\pi}^j \not\vDash p\right) \land
    \left(\exists k \geq j.(\tilde{\pi}^k \not\vDash q) \land
      \sum_{l=j}^{k-1} r_l \leq b\right).\tag{1}
  \end{equation*}
By   Lemma~\ref{lem:always:paths} there exists a (unique) path $\tilde{\pi}$ in $\mathcal R$
satisfying formula (1) (where $\pi$ is replaced by $\tilde{\pi}$).
  Let $j\geq 0$ and $k\geq j$
  be the minimal indices satisfying (1). If $\tilde{\pi}^j \not\vDash q$ then we
  are finished. Now assume $\tilde{\pi}^j \vDash q$.
In $\tilde{t}_i$ the proposition $p$ is not satisfied implying that the clock
  status in $\tilde{t}_j$ is \verb|off| and the clock value is
  $0$. It is clear that the clock value in all states between $\tilde{t}_j$ and $\tilde{t}_k$ is at most the sum of the duration of the tick steps between
  $\tilde{t}_j$ and $\tilde{t}_k$; moreover, by assumption, for all $m \in \mathbb N$ with $j < m \leq k$ it holds $\sum_{l=j}^{m-1}r_l \leq b$. It follows
  that $\tilde{\pi}^l \not\vDash clock > b$ for all $j\leq l \leq
  k$. Hence $\tilde{\pi}^j\not\vDash p \lor (q \weakuntil (clock
  > b))$, and thus we have shown that $\tilde{\mathcal R},
  \tilde{L_\Pi}, \{\tilde{t_0}\} \not\vDash \always(p \lor (q \weakuntil (clock > b)))$.
\end{proof}

\subsection{Completeness and Termination}
\label{sec:CompleteAndTerminate}

The strength of Real-Time Maude is clearly the expressiveness and the
generality of the systems that can be specified, and moreover,
powerful analysis techniques by simulation of specifications. However,
the drawback of modeling in Real-Time Maude is the fact that, since we
are dealing with general classes of infinite-state real-time systems,
formal analyses are in general \emph{incomplete}, and sometimes even
\emph{unsound}. In Real-Time Maude, on the one hand, an analysis
method is called \emph{sound} if any counterexample found by this
method is a real counterexample in the system. On the other hand, an
analysis method is called \emph{complete} if the fact that no
counterexample is found using this method actually implies that no
such counterexample exists. For instance, the LTL model checking of a
formula $\phi$ is sound, if any counterexample found by the model
checker is a real counterexample in the system. LTL model checking of
a formula $\phi$ is complete, if the fact that the model checker
responds that the formula is satisfied, the formula is actually
satisfied by the system, i.e.~there exists no counterexample
falsifying $\phi$.

\paragraph{Sound and complete model checking of time-bounded formulas}\label{sec:completeness}
In \cite{DBLP:journals/entcs/OlveczkyM07a} {\"O}lveczky and Meseguer
have characterized easily checkable conditions for specifications in
Real-Time Maude which imply soundness and 
completeness of LTL model checking under
the maximal time sampling strategy. Given a real-time rewrite theory
$\mathcal R$, a labeling function $L_\Pi$ with $\Pi$ atomic
propositions, then model checking an LTL formula $\phi$ with the
maximal time sampling strategy is sound and complete, if
(1) $\mathcal R$ is \emph{time-robust}, and (2) all atomic propositions in $\Pi$ are \emph{tick-stabilizing}.
Time-robustness of real-time rewrite theory intuitively means that
time can either advance by any amount, by any amount up to and
including a specific point in time, or not at all (and this property
is not affected by advancing time unless we reach the specific time
bound in the second case), and instantaneous rules can only be applied
when the system has advanced time by the maximal possible amount.
The second condition for sound and complete model checking is that all
atomic propositions are tick-stabilizing which means that they do not
change arbitrarily during a maximal time step, more precisely, tick-stabilizing state propositions are allowed to change not at all
during a maximal time step, or only once. For exact definitions see~\cite{DBLP:journals/entcs/OlveczkyM07a}.

As our goal is to achieve soundness and completeness of model checking generalized time-bounded response and time-bounded safety MTL formulas, it is essential that time-robustness is preserved by both $\eventually$- and $\always$-transformation.

\begin{theorem}\label{thm:time-robust}
  Let $\mathcal R$ be a real-time rewrite theory and let
  $\tilde{\mathcal
    R}$ be the result of the $\eventually$- or
  $\always$-transformation applied to $\mathcal R$.
  If $\mathcal R$ is time-robust, then $\tilde{\mathcal R}$ is time-robust.
\end{theorem}
\begin{proof}
  This assertion is proved by the observation that, according to
  Lemma~\ref{lem:bisimilar} and \ref{lem:always:paths}, both
  transformations do not change the original timed behavior of
  $\mathcal R$.
\end{proof} 

We will now sketch a proof that model checking generalized time-bounded response and time-bounded safety MTL formulas, with
tick-stabilizing atomic propositions, with the real-time system
specification for our case study as described in the previous sections
is indeed sound and complete.

%\begin{theorem}
%Given a component-based Real-Time Maude theory $\mathcal R$ of the form described in Section~\ref{sec:modeling} and a generalized time-bounded response MTL formula or a time-bounded safety MTL formula $\phi$ with tick-stabilizing atomic propositions.
%Then $\mathcal R$ and the transformed theory $\tilde{\mathcal R}$ are time-robust; the formula $\phi$ and its translation $\tilde{\phi}$ are tick-stabilizing.
%\end{theorem}
%

\begin{theorem}
Let $\mathcal R$ be a component-based real-time rewrite theory $\mathcal R$ of the form described in Section~\ref{sec:modeling} and let $\phi$ be a generalized time-bounded response MTL formula or a time-bounded safety MTL formula with tick-stabilizing atomic propositions.
Then time-unbounded model checking of the transformed formula $\tilde{\phi}$ w.r.t.\ the transformed theory $\tilde{\mathcal R}$ is sound and complete for the maximal time sampling strategy.
\end{theorem}
\begin{proof}
According to ~\cite{DBLP:journals/entcs/OlveczkyM07a} it is sufficient to prove that $\tilde{\mathcal R}$ is time-robust and $\tilde{\phi}$ only has tick-stabilizing atomic propositions.

A component-based real-time rewrite theory
$\mathcal R$ is time-robust since every instantaneous rewrite rule is
triggered by the expiration of a timer, or by the fact that the system is inconsistent, i.e.\ at least two
connected ports are not equal in value which can only happen
after a previous instantaneous step.\footnote{Note that our tick rule
  deviates from other well-known examples of object-oriented real-time
  rewrite theories. Beside the condition that \texttt{mte} returns a
  value greater than $0$ we require that the system is consistent,
  i.e.~any connected ports are equal in value.} According to Theorem~\ref{thm:time-robust}, our
transformations described in Sect.~\ref{sec:analysis} preserve
time-robustness, so the transformed theory $\tilde{\mathcal R}$ is time-robust as well.
%The transformations add clock objects to the global
%system state which do not affect the functions \verb|mte| and
%\verb|consistent|, and the conditions of instantaneous rules
%determining when such a rule can be applied are not changed (and no
%instantaneous rules are added). Hence a transformed real-time rewrite
%theory $\tilde{\mathcal R}$ remains time-robust given that $\mathcal R$ has
%been time-robust (which is the case for our specification of the
%digital advertising case study).

The second condition requires that all atomic propositions are
tick-stabilizing. By Lemma~\ref{lem:bisimilar} and \ref{lem:always:paths}, both transformations do not change the original time behavior of $\mathcal R$ hence all atomic propositions remain tick-stabilizing. Note that both transformations introduce a new (parameterized) atomic proposition \verb|clockLeq| which, however, is tick-stabilizing, since the truth of \texttt{clockLeq($b$)} for $b$ a time bound is not changed during a maximal time step, or only once.
%Since $\mathcal R$ is time-robust and an MTL formula $\phi$ with tick-stabilizing atomic propositions,
% the transformation yields a LTL formula
%$\tilde{\phi}$ which involves new atomic \emph{clocked} propositions
%of the form $clock \geq r$. However, these added clocked
%propositions are tick-stabilizing since during a maximal time step,
%the change their truth not at all, or once. Note that the clock is
%either not increased (if it is turned off or if the clock value is
%\verb|INF|), or it is increased and the satisfaction of $clock \geq
%r$ can only be changed once.
%It follows that for a time-robust real-time rewrite theory with
%tick-stabilizing propositions, model checking MTL formulas of the form
%$\always(p \lor \always_{\leq r} q)$ and $\always(\bigvee_{i\in
%  I}(\eventually_{\leq b_i} q_i))$ with the maximal time sampling
%strategy is sound and complete.
\end{proof}

\paragraph{Termination}
\label{sec:Termination}
In general, real-time rewrite theories are infinite-state systems for
which model checking will not terminate. However, if we are dealing
with finite-state systems, model checking will terminate. More
precisely, in a real-time rewrite theory $\mathcal R$ with a fixed
time sampling strategy, if both the reachable state space of $\mathcal
R$ from an initial state $\{t_0\}$ and the number of different rewrite
durations in all possible paths in $\mathcal R$ from $\{t_0\}$ are
finite, MTL model checking (of generalized time-bounded response MTL formulas or of time-bounded safety MTL formulas)
terminates. So
if the reachable state space in $\mathcal R$ from an initial state
$\{t_0\}$ (under a fixed time sampling strategy) is finite, then the
reachable state space of the transformed real-time rewrite theory
$\tilde{\mathcal R}$ is finite. The main point in the proof of this fact is
that the clock value is never increased more than necessary: if it exceeds the upper bound ($b_{max}$, $b$, resp.) then it is not increased any more which does not change the truth of propositions of the form $clock \geq b$, and ensures that the state space remains finite. For a detailed proof of this fact for a slightly different transformation (which however follows the same schema) we refer the interested reader to the work of Lepri et al.~\cite{rtrts-Olveczky}.
Thus, in our case study, model checking generalized time-bounded response MTL formulas or time-bounded safety MTL formulas with the maximal time sampling
strategy will terminate.

% \begin{lemma}
% If $\mathcal R$ is a time-robust real-time rewrite theory, then the transformed real-time rewrite theory $\tilde{\mathcal R}$ is time-robust.
% \end{lemma}

% \begin{theorem}[cf.~\cite{DBLP:journals/entcs/OlveczkyM07a}]
%   Let $\mathcal R$ be a time-robust real-time rewrite theory, $L_\Pi$ a labeling
%   function for $\mathcal R$ with $q_i \in \Pi$ for all $i \in I$, and
%   $\{t_0\}$ an initial state of $\mathcal R$. Let all $p \in \Pi$ be tick-stabilizing atomic propositions. Let $\tilde{\mathcal
%     R}$, $\tilde{L_\Pi}$, and $\{\tilde{t_0}\}$ be the result of the
%   $\always$-transformation applied to $\mathcal R$, $L_\Pi$, and
%   $\{t_0\}$. Then the following equivalence holds:
% \end{theorem}

% If the state space in $\mathcal R$ is finite, then the state space in $\tilde{\mathcal R}$ is finite

\subsection{Model Checking the Requirements of Digital Advertising}
In this section, we briefly describe the analysis of our real-time
specification of the digital advertising scenario in Real-Time Maude
using the untimed LTL model checking command. The analysis has been
performed on a single core processor (3.2GHz
Intel\textsuperscript{\textregistered} Pentium 4) with 2 GB of RAM.

Note that the transformations described above require that the real-time
object-oriented specifications are applied to are \emph{flat}
specifications in which rewrites happen only in the ``outermost''
configuration, and no rewrite is possible for attribute values. Our 
real-time specification, however, is \emph{non-flat}
as we are dealing with arbitrarily nested, hierarchical
components. A simple solution to this problem is to adapt all
rewrite rules such that they can only be applied at the outermost
layer. For hierarchical components, this implies that the transmission rule
must be duplicated for each layer of the component system (in our case study,
for two layers). This replication is part of the future work on automatizing
our analysis approach.

As all atomic propositions introduced in the following are tick-stabilizing
and moreover, the real-time specification of our case study is time-robust,
all analysis carried out (using the maximal time sampling strategy) are
complete, i.e.~if the model checking command of a temporal logic formula
returns a positive result, then the formula is provably correct for all timed
paths of the real-time specification.

Now, we discuss the analysis of our case study. To recall its basic
functionality, the digital advertising system can be found in one of two
configurations: in the first configuration, the system allows the user to
interact with the displayed content, while the system displays autoactive
content in the second configuration.

% We first verify that the system does not get stuck, i.e.~there is a
% always a state reachable such that time can advance. For model
% checking this property, we introduce an atomic proposition
% \verb|tick-enabled| which holds in a state for which the condition of
% the tick rule is satisfied, i.e.~the maximal time elapse is greater
% than zero and moreover, the configuration in that state is consistent.
% The property can be analyzed by the following command:
% \[\verb+(mc {initial} |=u [] <> tick-enabled .)+\]
% Furthermore, we can verify that the system always is either in
% configuration 1 or in configuration 2 (cf.~Sect.\ref{}, Fig.~\ref{})
% by the command
% \[\verb+(mc {initial} |=u [] ( in-C1 \/ in-C2 ) .)+\]
% where \verb|in-C1| and \verb|in-C2| are atomic proposition being
% satisfied if the connectors of configuration 1 and 2, respectively,
% are present in the assembly of the system component.

\paragraph{Verification of the guarantees (G1) and (G2).}
The contract to be satisfied by the digital advertising system consists of two
guarantees: (G1) Being an interactive ad, the system should react to a
user in front of the display. (G2) The content displayed must change
at least every ten seconds: an advertising campaign using a
large-scale display should not waste its capabilities by showing
static content.

Verifying the system guarantee (G2) amounts to model check that always
eventually the system changes the content of the display. This can be
model checked by the command
\[\small \verb+(mc {initial} |=u [] <> imgChange .)+\]
where \verb|imgChange| holds if the value of the provided port
\verb|ENV.imgChange| of the environment component is true. However,
the guarantee (G2) requires more: the displayed content must change at
least every ten seconds, so the above LTL formula is obviously
insufficient. Instead, the following formula must be used:
\[\always(\eventually_{\leq 10000}\ \textnormal{``image is
  changing''})\] This formula expresses that the image changes within
ten seconds regardless of the current system configuration. It is
worthwhile to investigate, since it is not immediately clear that the
system guarantees this properties in both configurations and under
arbitrary reconfigurations.  This can be checked by applying the
transformations presented above and executing the command
\[\small \verb+(mc {initial_MC} |=u [] <> ( imgChange /\ clockLeq(10000) ) .)+\]
where \verb|initial_MC| is the transformed initial state.
The model checking command took 8 minutes to complete, and did not find
any counterexamples.

To verify (G1), we check the property
\[\always(\always_{\leq 800}
\textnormal{``person is in front''} \to \eventually_{\leq 1000} \textnormal{``configuration 1''})\] using
the transformations specified above. The property states
that if a person stays in front of the display for at least $800$ milliseconds,
the system will be found in interactive mode within one second. Hence,
the property specifies that the system always guarantees to react to a person
in front of it. It can be model checked with the command
\[\small 
\begin{array}{l}
\verb+(mc {initial_MC} |=u [] ( (<> ( ~ persThereIn /\+\\
\verb+clockLeq(800))) \/ (<> ( in-C1 /\ clockLeq(1000)))) .)+
\end{array}\]
Model checking this property took 6 minutes, and again no counterexample
was found in the model of our case study.

\paragraph{Verification of state steadiness (G3).}
In addition to guaranteeing the system contract, we must assure that the
system cannot exhibit a behavior in which reconfigurations are continuously
performed and consume the available computing resources. In order to guarantee
that configuration states are reasonably stable, two properties (G3) are checked
with the help of the above transformations:
\[\always(\textnormal{``reconfiguration triggered in conf. 1''}
\to \always_{\leq 200} \textnormal{``configuration 1''}),\] and 
similarly for configuration 2.
%$\always(\mathit{reconfTriggeredInC2} \to \always_{\leq 200} C2)$.
These two
properties state that the reconfiguration does not happen instantaneously, but
must take at least 200 ms to complete. These properties guarantee that the system does not oscillate between configurations and that 
reconfigurations leave enough resources for the actual system operation.

The model checking command for the translated first property  is 
the following:
\[\small \verb+(mc {initial_MC} |=u [] ( (~ reconfTriggeredInC1) \/ (in-C1 W (~ clockLeq(200)))) .)+\]
Executing this model checking command took 12 minutes; the second property
can be translated and model-checked analogously.

Altogether, it is possible to check all real-time contract guarantees G1, G2,  and
G3 with the help of transformations and the built-in untimed Maude
LTL model-checker.

\section{Related Work}\label{sec:related}
Pervasiveness and ubiquity of software systems is a topic that has
been researched for about two
decades \cite{weiser99,satyanarayananm2001,leahu08}. A more recent
stream of research is focusing on leveraging the new sources
information becoming available through ubiquity of systems,
i.e. bio-signals. The ultimate goal is to create biocybernetic
loops~\cite{DBLP:conf/cec/SerbedzijaF09} in which the system and the
user create a feedback loop by influencing each other's reactions,
adapting the environment in an nonobtrusive way to the needs and ideas
of the user without requiring explicit interaction. Emotional
computing~\cite{DBLP:journals/ijmms/Picard03} is one of the most known
manifestations of this principle, but cognitive, and physical aspects
can be considered as well in the creation of a biocybernetic loop.

Constructing pervasive user-centric applications
\cite{advertising09:beyer-mayer-kroiss-schroeder} and, more general,
the construction of self-adaptive applications have been a field of
active research in recent years. In this work, we follow an approach
based on reconfiguration of the system in order to achieve
adaptability; \cite{DBLP:conf/woss/BradburyCDW04} gives a decent
overview of various approaches. Formal specification and verification
of component-based systems and their reconfiguration is presented in
several works, e.g.\ in \cite{BBB:TR06}, a logic-based approach to the
specification of reconfiguration is developed, and in
\cite{DBLP:journals/entcs/BarrosHM06}, reconfigurable components are
verified by model checking formulas of the $\mu$-calculus. However,
none of these frameworks for the verification of systems under
reconfiguration use time semantics and therefore, only untimed
properties can be verified.

Specifying systems with metric temporal logic goes back to the work of
Koyman in \cite{DBLP:journals/rts/Koymans90} and Hooman in
\cite{hoomanWidom89parle,hooman87parle}; in the latter work, a
compositional approach to the verification of system components with
metric temporal logic is presented. However, our work differs from the
above works by using a dynamic architecture instead of a static one.

In a previous work \cite{DBLP:conf/snpd/Olveczky08} on specification
and verification of systems in Real-Time Maude
\cite{DBLP:conf/maude/2007} already include ideas and methods how to
verify \emph{timed} temporal logic formulas using the LTL model
checker of Maude. However, the first automatized transformational
approach is presented in \cite{rtrts-Olveczky}, which cover MTL
formulas expressing the bounded response property or the minimum
separation property. In this work, we have extended the ideas of
\cite{rtrts-Olveczky} and presented analysis algorithms for two
further and more general classes of MTL formulas.

%Related work to pervasive user-centric applications,\\ component-based specifications and reconfiguration,\\ Real-time Maude, abstraction in Maude\\ formal approaches using metric temoral logic 

%see also related work of JLAP paper

%Numerous component models
%describe how exactly this is achieved, and how the components can
%communicate in order to achieve a common goal~\cite{LW07SCM}.

%\cite{DBLP:conf/mobilware/BeyerHKS09}

%\cite{rtrts-Olveczky}

\section{Concluding Remarks}\label{sec:concluding}
In the previous sections we have presented a new approach for formally
modeling and analyzing pervasive user-centric applicatons at an early
design stage.  A system is modeled as a set of components which
interact via connectors between provided and required ports. To allow
adaptation of the system to new situations, the system can be
dynamically reconfigured by changing the connections at runtime.

For specifying and prototyping such systems in a real-time setting, we
use the algebraic rewriting language Real-Time Maude.  Time-dependent
system properties are expressed in Metric Temporal Logic (MTL).
Real-Time Maude is also well-suited for model checking two practically
important classes of formulas, the so-called generalized time-bounded
response MTL formulas and the time-bounded safety MTL formulas.  By
extending the component-based Real-Time Maude models with suitable
clocks and by transforming these kinds of MTL formulas into pure LTL
formulas over the extended specification we have shown that these two
classes of formulas can be analyzed with the (untimed) Maude LTL model
checker, and that this analysis is sound, complete and terminating for
the maximal time sampling strategy.

As case study we have specified a simple adaptive advertising scenario
in Real-Time Maude and could automatically verify all three
requirements (G1--G3) with the Maude model checker by using our
analysis method. However, the execution of the model checking command
took in all cases several minutes although we had already abstracted
all values to boolean data. For more complex case studies, further
optimizations will be necessary to make model checking a practically
feasible analysis method.  One simple, but efficient technique is to
replace each model checking command, \texttt{mc} say, of form
%$modelCheck(\{\(t'\)REST\},\(q\sb{i}\))$
\texttt{modelCheck({t' REST}, q)} in a condition of a rule of the
extended theory $\tilde{\mathcal R}$ by a boolean expression; indeed,
each \texttt{q} is a state formula which can easily defined as a
boolean function \texttt{is-q} such that \texttt{is-q({t' REST})} is
true iff \texttt{mc} is true. Another technique is to reduce the
nondeterminism in hierarchical components by directly connecting the
ports of the environment with their corresponding ports of the
subcomponents (e.g. ENV.personThereIn with Camera.personThereIn). The
resulting specification, $\tilde{\tilde{\mathcal R}}$ say, is
stuttering equivalent (see e.g.~\cite{Mart�-oliet05theoroidalmaps})
with the original one; model checking $\tilde{\tilde{\mathcal R}}$ is
a matter of seconds, not of minutes.

The metric temporal logic properties in this paper take only
non-trivial upper bounds into account; the lower bound of any interval
is \texttt{0}. A ''natural'' extension of our work will be the study
of metric properties over intervals with non-null lower
bounds. Another interesting future work will be models with
time-dependent probabilistic behavior. Pervasive user-centric
applications interface with the real world through sensors and
actuators, which may be unreliable. With a probabilistic real-time
framework, it would be possible to model this uncertain behavior of
the environment, and reason about the performance of pervasive
user-centric applications in these environments.

\bibliographystyle{eptcs} % or whatever you prefer
\bibliography{rtrtsbib}

\end{document}